\documentclass[journal,twoside,web]{ieeecolor}
\usepackage{generic}
\usepackage{xcolor}
\usepackage{cite}
\usepackage{amsmath,amssymb,amsfonts,mathtools}
\usepackage{algorithmic}
\usepackage{graphicx}
\usepackage{textcomp}
\usepackage{subfigure}
\usepackage{caption}        
\usepackage{url}
\usepackage{tikz}
\usetikzlibrary{calc}
\usepackage[edges]{forest}
\usepackage{pgfplots}
\usepackage{hyperref}
\usetikzlibrary{shapes,arrows,automata}
\usetikzlibrary{shapes.multipart}
\DeclareMathAlphabet{\pazocal}{OMS}{zplm}{m}{n}
\usepackage{calrsfs}
\usepackage{dsfont}
\usepackage{stackengine}
\usepackage{xcolor}
\definecolor{GG}{rgb}{0,0,1}
\def\BibTeX{{\rm B\kern-.05em{\sc i\kern-.025em b}\kern-.08em
		T\kern-.1667em\lower.7ex\hbox{E}\kern-.125emX}}
\newtheorem{definition}{\textbf{Definition}}{\normalfont}{\normalfont}
{\normalfont}{\normalfont}
\newtheorem{remark}{Remark}{\normalfont}{\normalfont}
\newtheorem{theorem}{Theorem}
\newtheorem{assumption}{Assumption}
\newtheorem{proposition}{Proposition}
\newtheorem{lemma}{Lemma}

\newcommand{\figurename}{Fig.}
\newcommand{\tablename}{Table}

\newcommand{\smallmat}[1]{\left[ \begin{smallmatrix}#1 \end{smallmatrix} \right]}

\tikzset{
	curve box/.style={white, draw=black, text=black, rectangle, rounded corners},
}
\usetikzlibrary{decorations.pathreplacing}

\begin{document}
	\title{Beyond Shrinkage: Foundations of Data-Driven Control for Piecewise Affine Systems}
	\author{G. Giacomelli, \IEEEmembership{Student Member, IEEE}, V. G. Lopez, \IEEEmembership{Member, IEEE}, S. Formentin, \IEEEmembership{Senior Member, IEEE}, M.A. M\"uller, \IEEEmembership{Senior Member, IEEE}, V. Breschi, \IEEEmembership{Member, IEEE} 
		\thanks{Gianluca Giacomelli and Valentina Breschi are with Control Systems Group, Eindhoven University of Technology, 5612AZ Eindhoven, The Netherlands, (e-mails: \textsl{\{g.giacomelli, v.breschi\}@tue.nl}). Victor G. Lopez and Matthias A. M{\"u}ller are with the Institute of Automatic Control, Leibniz University Hannover, 30167 Hannover, Germany (e-mails: \textsl{\{lopez, mueller\}@irt.uni-hannover.de}). Simone Formentin is with the Dipartimento di Elettronica, Informazione e Bioingegneria,
		Politecnico di Milano, 20133 Milan, Italy (e-mail: \textsl{simone.formentin@polimi.it}).}
	}
	
	\maketitle
	
	\begin{abstract}
    Data-enabled predictive control (DeePC) has recently attracted attention as a promising approach for controlling systems directly from raw data, without requiring an explicit identification step. However, DeePC has not yet been extended to piecewise affine (PWA) systems, despite their extensive use in the (predictive) control literature and their universal approximation capabilities. To address this gap, in this work, we lay the foundations for data-enabled predictive control of PWA systems, providing: $(i)$ their behavioral characterization; $(ii)$ an extension of Willems' Fundamental Lemma to represent their behavior from raw data; $(iii)$ an analysis of the coherence of DeePC strategies using a linear predictor and shrinkage regularizers; and $(iv)$ a study of the impact of misclassification errors on structuring data for prediction. Our theoretical findings are validated by numerical results on a simple example, emphasizing the need to extend beyond a regularized version of the foundational DeePC framework to design control actions that are both effective and coherent with a PWA system’s behavior, thus ensuring the controller’s explainability. 
	\end{abstract}
	
	
	\section{Introduction}
	\IEEEPARstart{D}{ata-driven} control strategies are becoming increasingly popular due to the ever-growing availability of data and the increasing difficulty of first-principled modeling of complex systems. Among these approaches, the seminal works in \cite{coulson2019data,berberich2020data} have concurrently introduced similar strategies that replace the model-based predictor used in Model Predictive Control (MPC) with a predictor built from a raw batch of input/output trajectories of the system, thereby eliminating the need for an explicit identification step. This approach, connected to subspace identification as discussed in \cite{dorfler2022bridging,BRESCHI2023110961}, thus removes the operation that is often the most costly and time-demanding within an advanced control pipeline \cite{hjalmarsson2005experiment}. While this has led to significant progress, research has mainly focused on Linear-Time Invariant (LTI) systems, whereas extending this framework to nonlinear systems remains an active area of research (see the discussion in \cite{berberich2025overview}). In response to these challenges, several data-driven control schemes tailored to specific classes of nonlinear systems have been recently proposed, often relying on linearization of the underlying systems. Examples range from the approach proposed in \cite{verhoek2025behavioral} for Linear Parameter Varying (LPV) systems, which use measured or estimated scheduling signals (and embed them in the data-driven predictor) to map nonlinearities, to strategies using the linear-in-control input Koopman framework to tackle nonlinearities~\cite{de2024koopman}, nevertheless characterized by a prediction error that increases with the prediction horizon \cite{iacob2024koopman}. A different focus is taken in \cite{alsalti2023data}, which, for the class of feedback-linearizable systems, relies on a basis-function approximation to decompose the unknown nonlinearities. 

Building on the previous discussion of nonlinear systems, other strategies exploit existing data-driven predictive control schemes and use data-selection techniques to extend them to nonlinear control. The first hint of this kind appears in the seminal work on data-enabled predictive control (DeePC), which suggests using Lasso regularization to tackle nonlinearities within a certain distance of where batch data were collected \cite{coulson2019regularized}. To move beyond the limitation of a neighborhood of a single linearization point, \cite{naf2025choose,beerwerth2025less,li2025datamodel,pasini2026taming} propose different strategies to construct a data-driven model at each control instant by selecting previously stored data. Specifically, at a given time step, \cite{naf2025choose} uses manifold learning to embed the data, then selects data trajectories that are close to those predicted at the previous time step. Meanwhile, \cite{beerwerth2025less} proposes a contextual selection approach that selects data based on the Euclidean distances between the initial output condition and the reference, and between the data used to reconstruct them, respectively. Instead, \cite{li2025datamodel} proposes selecting data points based on their influence on the closed-loop cost via a subset-to-performance surrogate map, thereby introducing a model for the tracking loss. In another direction, \cite{pasini2026taming} builds on the framework proposed in \cite{chiuso2025harnessing} and exploits data-driven local affine predictors obtained via local linear regression \cite{rosolia2019learning} along the linearization trajectory. Alternatively, a continuously updated data-driven predictor has also been proposed by \cite{berberich2022linear}. However, it requires collecting new, persistently exciting data online, which is known to be challenging. 

Expanding on the aforementioned classes of systems, it is worth noting that, although some works have also considered data-driven control for switching systems (see, e.g., \cite{rotulo2022online, eising2024data,bianchi2025data,wang2025online,li2026online}), no strategy has yet been specifically devised for data-driven predictive control of piecewise affine (PWA) systems except for output-feedback controller synthesis~\cite{hu2026data}. 
Nonetheless, this class of systems has two main potential advantages. First, it is well known that PWA models are universal approximators \cite{lin2002canonical,breiman2002hinging}, making techniques developed for their control also applicable to other classes of nonlinear systems. Second, PWA systems feature a set of affine local models and a switching logic dictated by a polyhedral partition of the input/state (or input/output, depending on the representation) space, characteristics that are in line with the strategies relying on data selection and local linearizations \cite{naf2025choose,beerwerth2025less,pasini2026taming}. 
\paragraph*{Contribution}
        Building on the seminal work \cite{paoletti2009input}, we propose a behavioral representation for PWA systems, showing that PWA behaviors can be described via an LTI embedding under sparsity constraints dictated by the active mode. We further manipulate this representation for the local affine nature of the PWA systems to become explicit, ultimately leading us to the formulation of an extension of Willem's fundamental Lemma for the PWA system class. This result allows us to define a data-driven predictor for the PWA class, which relies on Mosaic data matrices (mirroring the local behaviors) and a binary selection matrix. Without enforcing the mixed-integer constraints characterizing our behavioral predictor and only leveraging the predictor's data structure, we then consider two extensions of DeePC~\cite{coulson2019data} with sparsity-promoting regularization and assess whether the achieved sparsity patterns could match the ones of the actual system behavior. In particular, apart from the elastic-net regularizer already evaluated in \cite{giacomelli2025insights} for nonlinear systems, we analyze for the first time the impact of a regularization scheme, the Composite Absolute Penalty (CAP) \cite{zhao2006grouped}, promoting the sparsification of groups of local behaviors. Both formally and via a simple numerical example, we show that using sparsification without explicitly optimizing for the active mode sequence eventually allows for input/output tracking but results in schemes whose outcomes are not coherent with the actual PWA behavior, paving the way for novel, hybrid DeePC schemes. At the same time, since the considered data structure relies on clustered data, we analyze the impact of misclassification errors when defining the local behaviors of the data-driven model, providing bounds on the attained costs with the considered shrinkage-based strategies.      
                %
                    %
                    %
	\paragraph*{Outline} Section \ref{sec:problem_formulation} introduces our setting and goal, including our assumption on the system to be controlled. In Section~\ref{sec:behaviors}, we then provide a characterization of PWA systems from a behavioral perspective, allowing us to define a data-driven predictor for a PWA system. Using only the data structures characterizing such a predictor, in Section~\ref{sec:shrinkage_based} we analyze the solution of two shrinkage-based DeePC schemes, one using elastic-net~\cite{giacomelli2025insights} and the other exploiting group-shrinkage regularization. As the considered data structures are built by clustering the available data, we analyze the impact of misclassification errors in Section~\ref{sec:misclassification}. Our theoretical findings are supported by the results of a numerical test in Section~\ref{sec:example}, followed by final remarks and directions for future work. 

	\paragraph*{Notation} The sets of integers, non-negative integers, and real numbers are denoted as $\mathbb{Z}$, $\mathbb{Z}_{\geq 0}$, and $\mathbb{R}$, respectively. Given a signal $\zeta_t \!\in\!\mathbb{R}^{n_{\zeta}}$, with $t\!\geq\! 0$, $q$ denotes the shift operator, i.e., $q^{-i}\zeta_t\!=\!\zeta_{t-i}$, for $i\! \in\! \mathbb{Z}$, while ${\zeta_{\left[\tau, \tau'\right]} \!=\!\begin{bmatrix} \zeta_{\tau}^{\top\!} & \!\cdots\! & \zeta_{\tau'}^{\top} \end{bmatrix}^{\top\!}}$, with $\tau,\tau'\!\in\!\mathbb{Z}_{\geq 0},\tau'\!>\!\tau$. We denote with $\mathrm{bdiag}(\zeta_{[\tau,\tau']})$ a block-diagonal matrix featuring the elements of $\zeta_{[\tau,\tau']}$ in its diagonal, while we indicate the Hankel matrix of depth $\ell>0$ associated with $\zeta_{[\tau,\tau']}$, namely $\pazocal{H}_\ell(\zeta_{\left[\tau, \tau'\right]})\in\mathbb{R}^{n_\zeta \ell \times \left( N_\zeta-\ell+1 \right)}$, as
\begin{equation}\label{eq:Hankel}
    \pazocal{H}_\ell(\zeta_{\left[\tau, \tau'\right]}) := \begin{bmatrix}
                     \zeta_\tau & \zeta_{\tau+1} & \cdots & \zeta_{\tau'-\ell+1}\\
                     \zeta_{\tau+1} & \zeta_{\tau+2} & \cdots & \zeta_{\tau'-\ell}\\
                     \vdots & \vdots & \ddots & \vdots\\
                     \zeta_{\tau+\ell-1} & \zeta_{\tau+\ell} & \cdots & \zeta_{\tau'}
\end{bmatrix}.
\end{equation}
Given a matrix $X \!\in\! \mathbb{R}^{m \times n}$, its Moore-Penrose inverse is $X^\dagger \!= \!(X^\top X)^{-1} X^\top$ while $\mathrm{vec}(X)$ is its vectorization. Meanwhile, $[X]_{i:j}$ represents the submatrix of $X$ from the $i$-th to the $j$-th row. Zero matrices are indicated as $\mathbf{0}$ (without specifying their dimensions), while unitary matrices are denoted as $\boldsymbol{1}_{m,n}\!\in\! \mathbb{R}^{m\times n}$. Given two matrices $X \!\in\! \mathbb{R}^{m \times n}$ and $Z \!\in\! \mathbb{R}^{p\times r}$, $X \!\otimes\! Z$ indicates the Kronecker product between the two matrices,  while for $m\!=\!p$ and $n\!=\!r$, $X\! \odot\! Z$ is the Hadamard one. An index set is denoted as $[n]=\{1,\dots,n\}$, with $n\in\mathbb{Z}_{>0}$. 

\section{Setting \& Goal}\label{sec:problem_formulation}
Consider a piecewise affine (PWA) system, whose dynamics are described by the following state-space model: 
	\begin{subequations}\label{eq:PWA_sys}
		\begin{equation}
			\begin{aligned}
				x_{t+1} &= A_{\sigma_t} x_t + B_{\sigma_t} u_t + f_{\sigma_t},\\
				y_t &= C_{\sigma_t} x_t + D_{\sigma_t} u_t + g_{\sigma_t},
			\end{aligned}
		\end{equation}
		where $x_t \in \mathbb{R}^{n_x}$ is the system's state at time $t\in\mathbb{Z}_{\geq0}$, $u_t \in \mathbb{R}^{n_u}$ and $y_t \in \mathbb{R}^{n_y}$ are its input and output, respectively, while, $\sigma_t \in [n_\sigma]$ indicates its \emph{active mode}. The latter is driven by a \emph{complete} polyhedral partition of the state-input space, i.e., 
		\begin{equation}\label{eq:active_mode}
			\sigma_t\!=\!i \!\iff\! \begin{bmatrix} x_t \\  u_t \end{bmatrix} \!\in\! \pazocal{P}_i^{\mathrm{SS}\!}=\!\left\{ \begin{bmatrix} x_t \\ u_t \end{bmatrix}
			~ \mathrm{s.t.}~ P_i^{\mathrm{SS}}\begin{bmatrix}
				x_t\\
				u_t\\
				1
			\end{bmatrix} \!\preceq\! \mathbf{0} \right\}\!,
		\end{equation}
		with $\bigcup_{i=1}^{n_\sigma}\pazocal{P}_{i}^{\mathrm{SS}}=\mathbb{R}^{n_u+n_x}$ and $\overset{\circ}{\pazocal{P}^{\mathrm{SS}}_i} \cap \overset{\circ}{\pazocal{P}^{\mathrm{SS}}_j} = \varnothing$ for all $i\neq j$, with $i,j \in [n_\sigma]$ and $\overset{\circ}{\pazocal{P}_i^{\mathrm{SS}}}$ denoting the interior of the $i$-th polyhedron.
	\end{subequations}	
	Let us assume that this system is \emph{minimal}, both with respect to the dimension of the state and the number of regions, \emph{open-loop stable}, \emph{finite-time observable}, and \emph{controllable} according to the definitions reported in Appendix~\ref{appendix:useful_notions}. 
	
	Suppose that we wish to design a (predictive) controller for such a system, yet we have no access to the matrices $\{(A_i,B_i,C_i,D_i,f_i,g_i)\}_{i=1}^{n_\sigma}$ characterizing the local dynamics of the PWA system, as well as the polyhedral partition of the state-input space dictating the active mode. Meanwhile, assume that we have collected a trajectory of the PWA system of finite length (see Definition~\ref{def:trajectory_PWA} in Appendix~\ref{appendix:useful_notions}) and, thus, we have access to a set of input and \emph{noise-free} output data
	\begin{equation}\label{eq:data}
		\pazocal{D}=\left\{u_{[1,N]}^d,y_{[1,N]}^d\right\}.
	\end{equation}
	
	Under this assumption, our \emph{first} goal is to characterize the behavior of a PWA system using only $\pazocal{D}$, establishing conditions on the \textquotedblleft richness\textquotedblright \ of these data to represent the PWA system. This data-based characterization then allows us to pursue our \emph{second} objective, which is to analyze theoretically and numerically the suitability of two shrinkage-based DeePC schemes for the predictive control of PWA systems.   
	
	\begin{remark}[Why noiseless data]
		Although real-world data are noisy, this preliminary work focuses on a noise-free setting to unveil challenges that can already be faced in controlling a PWA system directly from data in an idealized setting. Treatment of noisy data is postponed to future work.\hfill $\square$
	\end{remark}
	
	\section{PWA systems under behavioral lens}\label{sec:behaviors}
	As a first step, we transition from the state-space representation in \eqref{eq:PWA_sys} to an equivalent input-output description. To this end, apart from the assumption on \eqref{eq:PWA_sys} already introduced in Section~\ref{sec:problem_formulation}, let us consider two additional technical conditions.
	\begin{assumption}[Condition C1 in~\cite{paoletti2009input}]\label{ass:condition_C1}  
		Consider a pair of initial states $x_0\!\neq\! x_0'$, $x_0,x_0' \!\in\! \mathbb{R}^{n_x}$, and an input $\{u_t\}_{t\in \mathbb{Z}_{\geq 0}}$. Let $y_{[0,\ell-1]}$ and $y_{[0,\ell-1]}'$ be the associated outputs according to \eqref{eq:PWA_sys}. If $y_{\tau}\!\neq\! y_{\tau}'$ for all $\tau \!\in\! [0,\ell\!-\!1]$, then $
		y_{\ell}\!\neq\! y_{\ell}'$. \hfill $\square$ 
	\end{assumption} 
	\begin{assumption}[Condition C3 in~\cite{paoletti2009input}]\label{ass:condition_C3}
		For every admissible mode sequence $\sigma_{[0,\ell-1]}$ (see Definition~\ref{def:feasible_modes} in Appendix~\ref{appendix:useful_notions}) of \eqref{eq:PWA_sys}, with $\ell \in \mathbb{Z}_{> 0}$, there exists $\Xi \in \mathbb{R}^{n_y \times n_y (\ell-1)}$ such that the observability matrix
		\begin{equation}\label{eq:observability_matrix}
			\pazocal{O}_{0}^{\ell-2}\!=\!\begin{bmatrix}
            C_{\sigma_0}^{\top\!\!\!}\! &     		(C_{\sigma_{1}}A_{\sigma_0})^{\top\!\!\!} & \!\!\!\cdots\!\!\! &	(C_{\sigma_{\ell-2}}\Pi_{j=0}^{\ell-3}A_{\sigma_{\ell-3-j}})^{\!\top} 		
			\end{bmatrix}^{\top\!\!},\!\!
		\end{equation}
		satisfies $\Xi\pazocal{O}_{0}^{\ell-2}=C_{\sigma_{\ell-1}}\Pi_{j=0}^{\ell-2}A_{\sigma_{\ell-2-j}}$.
		\hfill $\square$ 
	\end{assumption}
	Based on~\cite[Corollary 1]{paoletti2009input}, finite-time observability together with Assumptions~\ref{ass:condition_C1}-\ref{ass:condition_C3} guarantee the existence of a PieceWise Affine model with eXogenous inputs (PWARX) equivalent to \eqref{eq:PWA_sys} according to the following definition (see instead Appendix~\ref{appendix:useful_notions} for the definition of trajectories of \eqref{eq:PWA_sys} and \eqref{eq:PWARX_sys}). 
	\begin{definition}[Equivalence of \eqref{eq:PWA_sys} and \eqref{eq:PWARX_sys}]\label{def:equivalence}
		The PWA representations given in \eqref{eq:PWA_sys} and \eqref{eq:PWARX_sys} are (input-output) equivalent if the set of input-output trajectories of \eqref{eq:PWA_sys} and \eqref{eq:PWARX_sys} coincide.
	\end{definition}
	
	Therefore, the PWA system is equivalently represented by 
	\begin{subequations}\label{eq:PWARX_sys}
		\begin{equation}
			y_t+\sum_{j=1}^{n_a}a_j(s_t)y_{t-j}=\sum_{j=0}^{n_b}b_j(s_t)u_{t-j}+c(s_t),
		\end{equation}
		where $s_t \in [S]$ denotes the active mode of the system at time $t \in \mathbb{Z}_{\geq 0}$, now dictated by  
		\begin{equation}\label{eq:mode_selection}
			s_t=i \iff \begin{bmatrix}
				y_{[t-n_a,t-1]}\\
				u_{[t-n_b,t]}
			\end{bmatrix} \in \pazocal{P}_{i}^{\mathrm{i/o}},
		\end{equation}
		where
		\begin{equation}\label{eq:PWARX_partition}
			\pazocal{P}_{i}^{\mathrm{i/o}\!}\!=\!\left\{\begin{bmatrix}
				y_{[t-n_a,t-1]}\\
				u_{[t-n_b ,t]}
			\end{bmatrix} \mbox{ s.t. } P_{i}^{\mathrm{i/o}\!\!}\begin{bmatrix}
				y_{[t-n_a ,t-1]}\\
				u_{[t-n_b,t]}\\
				1
			\end{bmatrix}\preceq \mathbf{0}\right\}\!,~~i \!\in\! [S],
		\end{equation}
		satisfy $\bigcup_{i=1}^{S}\pazocal{P}_{i}^{\mathrm{i/o}\!}=\mathbb{R}^{n_an_y+(n_b+1)n_u}$ and $\overset{\circ}{\pazocal{P}_{i}^{\mathrm{i/o}\!}} \cap \overset{\circ}{\pazocal{P}_{j}^{\mathrm{i/o}\!}}
        =\varnothing$ for all $i,j \in [S]$, with $i\neq j$. In what follows, we assume that shifting to a PWARX representation preserves controllability.
		\begin{assumption}[Controllability of \eqref{eq:PWARX_sys}]
			The PWARX system in \eqref{eq:PWARX_sys} is controllable.
		\end{assumption}
	\end{subequations}
	Moreover, we will always rely on the following hypothesis.
	\begin{assumption}[Known number of modes]\label{ass:known_numberS}
		The number of modes $S$ of the PWARX system in \eqref{eq:PWARX_sys} is known.
	\end{assumption}
     Note that this assumption is rather limiting in practice. Therefore, future work will focus on lifting it. 
	\begin{remark}[Complexity \eqref{eq:PWARX_sys}]
		As discussed in \cite[Section III.A]{paoletti2009input}, the number of possible modes $S$ of the PWARX representation in \eqref{eq:PWARX_sys} is generally larger than that of the equivalent PWA state-space representation in \eqref{eq:PWA_sys}, i.e., $S \geq n_{\sigma}$. \hfill $\square$  
	\end{remark}

	\subsection{Kernel representation \& behavior of the PWA system}
	By introducing the lag of the system $\rho=\mathrm{max}\{n_a,n_b\}$, the relationship in \eqref{eq:PWARX_sys} leads to the following kernel representation 
	\begin{align}\label{eq:kernel_representation}
		\nonumber &\sum_{j=0}^{\rho} r_j(
		s_{t})q^{-j}w_t=\mathbf{0}  \\
		&\Rightarrow \sum_{i=1}^{S}\left[\sum_{j=0}^{\rho} 
		r_j(i)q^{-j}w_t\right]\mathds{1}(s_t\!=\!i)=\mathbf{0},~~\forall t \in \mathbb{Z}_{\geq 0},
	\end{align}
	where $w_t\!=\!\mathrm{col}(y_t,u_t, 1)$, $q^{-1}$ is the back-shift operator, and $r_j(i)\!=\!\smallmat{a_{j}(i) & -b_j(i) & -c(i)} \!\in\! \mathbb{R}^{n_y\times (n_y+n_u+1)}$ satisfies $a_0(i)\!=\!I$ and $c_j(i)\!=\!\mathbf{0}$, for all $j\! \in\![1,\rho]$ and $i \!\in\! [S]$. Meanwhile, $\mathds{1}(s_t\!=\!i)$ is an indicator function such that $\mathds{1}(s_t\!=\!i)=1$ if $s_t\!=\!i$, and $\mathds{1}(s_t\!=\!i)=0$ otherwise.
    
	Starting from this, we can then define the behavior of the PWA system as
	\begin{equation}\label{eq:PWA_behavior}
		\mathcal{B}^{\mathrm{PWA}\!}\!=\!\left\{(w,\!s) \!\in\! (\mathbb{W}\!\!\times\!\! [S])^{\mathbb{Z}_{\geq 0}}\big|\eqref{eq:mode_selection}\mbox{-}\eqref{eq:PWARX_partition} \mbox{ and }\eqref{eq:kernel_representation} \mbox{ hold}\right\}\!,
	\end{equation}
	where $\mathbb{W}= \mathbb{R}^{n_y+n_u+1}$, as well as the truncated behavior over an horizon $L \in \mathbb{Z}_{\geq 0}$ as
	\begin{align}\label{eq:PWA_behavior_truncated}
		\nonumber \mathcal{B}_{|L}^{\mathrm{PWA}\!}\!=&\!\left\{(w,s) \!\in\! (\mathbb{W}\!\times\! [S])^{L}\big|\exists (\bar{w},\bar{s}) \in \mathcal{B}^{\mathrm{PWA}}  \right.\\
		& \qquad ~ \left. \mbox{ s.t. } w_t=\bar{w}_{t}, s_t=\bar{s}_t, \forall t \in [0,L\!-\!1]\right\}.
	\end{align}
	Moreover, we can define the \emph{projected} and truncated {projected mode behaviors} as
	\begin{align}
		&\mathcal{B}^{\mathrm{PWA}\!}_{[S]}\!=\!\left\{s \!\in\! [S]^{\mathbb{Z}_{\geq 0}}| \exists w \!\in\! \mathbb{W}^{\mathbb{Z}_{\geq 0}} \mbox{ s.t. } (w,s) \!\in\! \mathcal{B}^{\mathrm{PWA}}\right\},\label{eq:projected_behavior}\\
		&\mathcal{B}^{\mathrm{PWA}\!}_{[S]|L}\!=\!\!\left\{s \!\!\in\! [S]^{L}| \exists \bar{s} \!\in\! \mathcal{B}_{[S]}^{\mathrm{PWA}\!\!} \mbox{ s.t. }\! s_t\!\!=\!\!\bar{s}_t, \forall t \!\!\in\! [0,L\!\!-\!1]\right\}\!,\label{eq:truncated_projected_behavior}
	\end{align}
	where the latter can be seen as \textquotedblleft behavioral translations\textquotedblright \ for the concept of \emph{feasible} (infinite and finite) {mode sequence}, respectively, with the latter defined as follows.
	\begin{definition}[PWARX feasible mode sequence]\label{def:feasible_modes2}
		Given $\ell \in \mathbb{Z}_{> 0}$, the mode sequence $s_{[0,\ell-1]}' \in [S]^{\ell}$ is feasible for \eqref{eq:PWARX_sys} if there exists an input/output sequence $\{u_{[0,\ell-1]},y_{[0,\ell-2]}\}$ such that $s_{[0,\ell-1]}=s_{[0,\ell-1]}'$.  
	\end{definition}
        In addition, we can introduce the behavior compatible with a fixed, feasible mode sequence $s \in [S]^{\mathbb{Z}_{\geq 0}}$ as
	\begin{equation}\label{eq:fixed_s_behavior}
		\mathcal{B}^{\mathrm{PWA}}_{s}=\left\{w \in \mathbb{W}^{\mathbb{Z}_{\geq 0}} \big| (w,s) \in \mathcal{B}^{\mathrm{PWA}}\right\},
	\end{equation}
	and its truncated counterpart as
	\begin{equation}\label{eq:fixed_s_behavior_truncated}
		\mathcal{B}^{\mathrm{PWA}\!}_{s|L}\!=\!\left\{w \!\in\! \mathbb{W}^{L} \big| \exists \bar{w} \!\in\! \mathcal{B}^{\mathrm{PWA}\!}_{s} \mbox{ s.t. } w_t\!=\!\bar{w}_t, \forall t \!\in\! [0,L\!-\!1]\right\}.
	\end{equation}
    The equivalence of \eqref{eq:PWA_sys} and \eqref{eq:PWARX_sys} with \eqref{eq:kernel_representation} now allows us to characterize the dimension of $\mathcal{B}^{\mathrm{PWA}\!}_{s|L}$ as follows.
    \begin{lemma}[Dimension of $\mathcal{B}^{\mathrm{PWA}\!}_{s|L}$]\label{lemma:behaviour_dimension}
        Let $L\geq \rho$ and let the state-space PWA system \eqref{eq:PWA_sys} equivalent to \eqref{eq:PWARX_sys} be observable in $L$ steps. Given any feasible sequence $s_{[0,L-1]} \in \mathcal{B}^{\mathrm{PWA}\!}_{[S]|L}$, then the dimension of $\mathcal{B}^{\mathrm{PWA}\!}_{s|L}$ satisfies
        \begin{equation}\label{eq:dimension_behavior}
            \mathrm{dim}(\mathcal{B}^{\mathrm{PWA}\!}_{s|L})=n_x+n_uL +1.
        \end{equation}
    \end{lemma}
    \begin{proof}
        The proof can be found in Appendix~\ref{appendix:proof_dimension_behavior}.
    \end{proof}    
      \begin{remark}[Lag and partition]
        In the remainder of the paper, we will always consider the lag $\rho$ rather than $n_a$ and $n_b$. In turn, this implies that the active mode will be dictated by the equivalent condition
        \begin{subequations}\label{eq:mode_selection2}
        \begin{equation}\label{eq:PWARX_partition1}
        s_t =i \iff \begin{bmatrix}
            y_{[t-\rho,t-1]}\\
            u_{[t-\rho,t]}
        \end{bmatrix} \in \pazocal{P}_i
        \end{equation}
        where
        \begin{equation}\label{eq:PWARX_partition2}
			\pazocal{P}_{i}\!=\!\left\{\!\begin{bmatrix}
				y_{[t-\rho,t-1]}\\
				u_{[t-\rho ,t]}
			\end{bmatrix} \mbox{ s.t. } P_{i}\begin{bmatrix}
				y_{[t-\rho ,t-1]}\\
				u_{[t-\rho,t]}\\
				1
			\end{bmatrix}\!\preceq\! \mathbf{0}\!\right\}\!,~i \!\in\! [S],\!
		\end{equation}
        where $P_{i}$ is similar to $P_{i}^{\mathrm{i/o}}$, yet potentially featuring some rows of zeros. 
        \end{subequations}
    \end{remark}
    \subsection{PWA behaviors from data without structural insights}
    Let us now compactly recast the kernel in \eqref{eq:kernel_representation} as follows
    \begin{subequations}\label{eq:kernel_equivalent2}
		\begin{equation}
			\sum_{j=0}^{\rho} r_jq^{-j}\tilde{w}_t^{e}=\mathbf{0},
		\end{equation}  
		with 
		\begin{equation}\label{eq:stacked_coefficients}
			r_j\!=\!\begin{bmatrix}
				r_j(1) & r_j(2) & \ldots & r_j(S) 
			\end{bmatrix} \in \mathbb{R}^{n_y\times S(n_y+n_u+1)},
		\end{equation}
		and $\tilde{w}_{t}^{e}=\mathrm{col}(\tilde{y}^{e}_{t},\tilde{u}_{t}^{e}) \in \mathbb{R}^{S(n_y+n_u+1)}$ where we define 
		\begin{equation}\label{eq:useful_definitions}
        \begin{aligned}
			&\tilde{y}^{e}_{t-j}\!=\! \mathds{1}(s_t)  \otimes y_{t-j} \in \mathbb{R}^{Sn_y},\\ &\tilde{u}_{t-j}^{e}\!=\! \begin{bmatrix}
			    \mathds{1}(s_t)  \!\otimes\! u_{t-j} \\ \mathds{1}(s_t)
			\end{bmatrix}\!\!\in \mathbb{R}^{S(n_{u\!}+1)},
            \end{aligned}
		\end{equation}
	\end{subequations}
	 and $\mathds{1}(s_t)\!\!=\!\!\smallmat{\!
		\mathds{1}(s_t=1) &
		\mathds{1}(s_t=2) &\!
		\cdots \!&
		\mathds{1}(s_t=S)\!}^{\top\!}$, for all $j \!\in\! [0,\rho]$ and $\forall t \in \mathbb{Z}_{\geq 0}$. Similar to what is done in \cite[Section VI.E]{faye2025willems}, this reformulation allows us to obtain a \emph{linear embedding} of the PWA system and to define (see \cite[Section VI]{willems1986time}) the following \emph{linear}, \emph{complete}, and \emph{shift-invariant} behavior:
	\begin{equation}\label{eq:tilde_A_behavior}
		\tilde{\mathcal{B}}^{e}=\left\{\tilde{w}^{e} \in \tilde{\mathbb{W}}^{ \mathbb{Z}_{\geq 0}}\big| R(q)\tilde{w}^{e}=\mathbf{0} \right\},
	\end{equation}
	where $\tilde{\mathbb{W}}\!=\!\mathbb{R}^{S(n_y+n_u+1)}$, $R(q)\!=\!\sum_{j=0}^{\rho}r_jq^{-j}$ and $r_j$ defined in \eqref{eq:stacked_coefficients}. Note that, as given by a linear embedding, $\tilde{\mathcal{B}}^{e}$ disregards the interdependence between $\tilde{u}_{t}^{e}$, $\tilde{y}_{t}^{e}$, and $s_{t}$. Based on these definitions and given a feasible mode sequence $s \!\in\! [S]^{\mathrm{Z}_{\geq 0}}$, let $\Pi_{y}$ be the projection of $\tilde{w}^{e}$ onto $y \!\in\! \mathcal{B}_{s}^{\mathrm{PWA}}$ and $\Pi_{u}$ be the projection of $\tilde{w}^{e}$ onto $\smallmat{u^{\!\top} & 1}^{\!\top\!} \!\in\! \mathcal{B}_{s}^{\mathrm{PWA}}$, and, let $\Pi_{\tilde{y}^{e}}$ and $\Pi_{\tilde{u}^{e}}$ be the projections of $\tilde{w}^{e}$ onto $\tilde{y}^{e} \in \tilde{\mathcal{B}}^{e}$ and $\tilde{u}^{e} \in \tilde{\mathcal{B}}^{e}$, respectively. In the same spirit as \cite[Lemma 11]{markovsky2022data}, we can formalize the relation between $\tilde{\mathcal{B}}^{e}$ and $\mathcal{B}^{\mathrm{PWA}}_{s}$ as follows.
    \begin{lemma}[LTI embedding of PWA behaviors]
    The behavior $\tilde{\mathcal{B}}^{e}$ in \eqref{eq:tilde_A_behavior} embeds $\mathcal{B}^{\mathrm{PWA}}_{s}$ in \eqref{eq:PWA_behavior}, i.e., $\mathcal{B}_{s}^{\mathrm{PWA}} \subseteq \Pi_{w} \tilde{\mathcal{B}}^{e}$. Such an embedding is exact, namely $\mathcal{B}_{s}^{\mathrm{PWA}}=\tilde{\mathcal{B}}^{e}$, when the constraints    
\begin{equation}\label{eq:constraints_for_equality}
    \Pi_{\tilde{y}^{e}}\tilde{w}_{t-j\!}^{e}\!=\!\mathds{1}(s_t) \!\otimes\!  \Pi_{y}\tilde{w}_{t-j\!}^{e},~\Pi_{\tilde{u}^{e}}\tilde{w}_{t-j\!}^{e}\!=\!\mathds{1}(s_t) \!\otimes\!  \Pi_{u}\tilde{w}_{t-j\!}^{e},
 \end{equation}
    are imposed for all $j \in [0,\rho]$ and $t \in \mathbb{Z}_{\geq 0}$, with $s_t$ verifying \eqref{eq:mode_selection2}. Therefore, $\mathcal{B}_{s}^{\mathrm{PWA}}$ can be equivalently represented as
    \begin{align}\label{eq:PWA_behvaior_constr}
        \nonumber \mathcal{B}_{s}^{\mathrm{PWA}\!\!}=\!\{\Pi_{w}\tilde{w}^{e}\big|&\tilde{w}^{e}\!\in\! \tilde{\mathcal{B}}^{e}\mbox{ and \eqref{eq:constraints_for_equality} holds } \forall j \in [0,\rho],\\
        & \qquad \forall s_t \mbox{ satisfying \eqref{eq:mode_selection2}, } t \in \mathbb{Z}_{\geq 0}\},
    \end{align}
    where $\Pi_{w}$ is the projection of $\tilde{w}^{e}$ onto $w$.
    \end{lemma}
    \begin{proof}
    	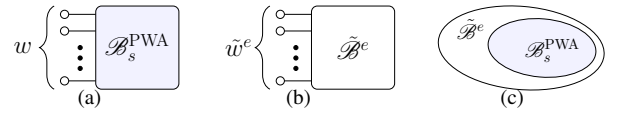
\begin{figure}[!tb]
		\centering
		\begin{tabular}{ccc}
            \subfigure[\label{Fig:PWA_behavior}]{\scalebox{.55}{\begin{tikzpicture}
						\node[coordinate] (start) {};
						\node[curve box, below of=start,node distance=0cm,minimum height=2cm,minimum width=2cm,fill=blue!5!white] (block1) {\LARGE{$\mathcal{B}^{\mathrm{PWA}}_{s}$}};
						\node[coordinate,left of=block1,node distance=1cm] (aid_1) {};
						\node[coordinate,above of=aid_1,node distance=.8cm] (Input_sys_1) {};
						\node[draw,circle,left of=Input_sys_1,node distance=.75cm,inner sep=.7mm, minimum size=1mm, fill=white] (Pin_1) {};
						\node[coordinate, above of=aid_1,node distance=.4cm] (Input_sys_2) {};
						\node[draw,circle,left of=Input_sys_2,node distance=.75cm,inner sep=.7mm, minimum size=1mm, fill=white] (Pin_2) {};
						\node[coordinate, below of=aid_1,node distance=.25cm] (aid_2) {};
						\node[draw,circle,left of=aid_2,node distance=.375cm,inner sep=0, minimum size=1mm, fill=black] (dot_mid) {};
						\node[draw,circle,above of=dot_mid,node distance=.25cm,inner sep=0, minimum size=1mm, fill=black] (dot_up) {};
						\node[draw,circle,below of=dot_mid,node distance=.25cm,inner sep=0, minimum size=1mm, fill=black] (dot_down) {};
						\node[coordinate,below of=aid_1,node distance=.8cm] (Input_sys_end) {};
						\node[draw,circle,left of=Input_sys_end,node distance=.75cm,inner sep=.7mm, minimum size=1mm, fill=white] (Pin_end) {};
						\node[coordinate, left of=aid_1, node distance=1cm] (aid_3) {};
						\node[coordinate, above of=aid_3, node distance=1cm] (graph_1) {};
						\node[coordinate, below of=aid_3, node distance=1cm] (graph_2) {};
						
						\draw[-]        (Pin_1)   -- (Input_sys_1);
						\draw[-]        (Pin_2)   -- (Input_sys_2);
						\draw[-]        (Pin_end)   -- (Input_sys_end);
						\draw [decorate,decoration={brace,amplitude=10pt,mirror},  line width=0.5pt] (graph_1)  -- (graph_2) node[midway,xshift=-.75cm]{\LARGE{$w$}};
			\end{tikzpicture}}}
			& 
			\subfigure[\label{Fig:tilde_A_behavior}]{
				\scalebox{.55}{\begin{tikzpicture}
						\node[coordinate] (start) {};
						\node[curve box, below of=start,node distance=0cm,minimum height=2cm,minimum width=2cm,fill=white] (block1) {\LARGE{$\tilde{\mathcal{B}}^{e}$}};
						\node[coordinate,left of=block1,node distance=1cm] (aid_1) {};
						\node[coordinate,above of=aid_1,node distance=.8cm] (Input_sys_1) {};
						\node[draw,circle,left of=Input_sys_1,node distance=.75cm,inner sep=.7mm, minimum size=1mm, fill=white] (Pin_1) {};
						\node[coordinate, above of=aid_1,node distance=.4cm] (Input_sys_2) {};
						\node[draw,circle,left of=Input_sys_2,node distance=.75cm,inner sep=.7mm, minimum size=1mm, fill=white] (Pin_2) {};
						\node[coordinate, below of=aid_1,node distance=.25cm] (aid_2) {};
						\node[draw,circle,left of=aid_2,node distance=.375cm,inner sep=0, minimum size=1mm, fill=black] (dot_mid) {};
						\node[draw,circle,above of=dot_mid,node distance=.25cm,inner sep=0, minimum size=1mm, fill=black] (dot_up) {};
						\node[draw,circle,below of=dot_mid,node distance=.25cm,inner sep=0, minimum size=1mm, fill=black] (dot_down) {};
						\node[coordinate,below of=aid_1,node distance=.8cm] (Input_sys_end) {};
						\node[draw,circle,left of=Input_sys_end,node distance=.75cm,inner sep=.7mm, minimum size=1mm, fill=white] (Pin_end) {};
						\node[coordinate, left of=aid_1, node distance=1cm] (aid_3) {};
						\node[coordinate, above of=aid_3, node distance=1cm] (graph_1) {};
						\node[coordinate, below of=aid_3, node distance=1cm] (graph_2) {};
						
						\draw[-]        (Pin_1)   -- (Input_sys_1);
						\draw[-]        (Pin_2)   -- (Input_sys_2);
						\draw[-]        (Pin_end)   -- (Input_sys_end);
						\draw [decorate,decoration={brace,amplitude=10pt,mirror},  line width=0.5pt] (graph_1)  -- (graph_2) node[midway,xshift=-.75cm]{\LARGE{$\tilde{w}^{e}$}};
			\end{tikzpicture}}}
			& 
			\subfigure[\label{Fig:behavior_set}]{
				\scalebox{.55}{\begin{tikzpicture}
						\node[coordinate] (start) {};
						\node[coordinate,right of=start,node distance=.5cm] (aid_1) {};
						
						\draw[rotate=-5] (start) ellipse (2cm and 1cm) node[midway,xshift=-1.25cm, yshift=.35cm, rotate=5]{\Large{$\tilde{\mathcal{B}}^{e}$}};
						\draw[fill = blue!5!white, rotate=-5] (aid_1) ellipse (1.3cm and .7cm) node[midway,xshift=.75cm, rotate=5]{\Large{$\mathcal{B}_{s}^{\mathrm{PWA}}$}};
						
			\end{tikzpicture}}}
		\end{tabular}
		\caption{PWA system behavior under a feasible mode sequence (a), its LTI embedding with mode-extended signals (b), and their relationship (c), inspired by \cite[Figure 2]{verhoek2025behavioral}.}\label{fig:data_structures}
	\end{figure}
        By construction (see \eqref{eq:useful_definitions}) $\tilde{w}^{e}_{t} \!\in\! \tilde{\mathcal{B}}^{e}$ embeds $w_{t}\!=\!\mathrm{col}(y_{t},u_{t},1)\! \in \!\mathcal{B}_{s}^{\mathrm{PWA}}$ for all $t\! \in\! \mathbb{Z}_{\geq 0}$. At the same time, $\tilde{w}^{e} \!\in\! \tilde{\mathbb{W}}$ might not match the sparsity pattern induced by the active mode $s_t$ at a given time instant, as $\tilde{\mathcal{B}}^{e}$ disregards the interdependence between $\tilde{u}_{t}^{e}$, $\tilde{y}_{t}^{e}$, and $s_{t}$ (in turn depending on $w$ according to \eqref{eq:PWARX_partition1}-\eqref{eq:PWARX_partition2}). Therefore, the embedding is not exact if $\tilde{w}_{t}^{e}$ does not match a sparsity pattern that is coherent with the active mode at time $t$. It is straightforward to verify that such a sparsity pattern is enforced through \eqref{eq:constraints_for_equality}, guaranteeing that the embedding is exact. 
    \end{proof}
    This result allows us to formally show that behavior \eqref{eq:tilde_A_behavior} over-approximates the PWA one, as exemplified in~\figurename{\ref{fig:data_structures}}. At the same time, it allows us to take advantage of existing results on data-driven LTI representations \cite{willems2005note,markovsky2022identifiability} to characterize finite length trajectories $\tilde{w}_{[0,L-1]}^{e} \!\in\! \tilde{\mathcal{B}}_{|L}$ from data. 
    
    Indeed, given a dataset $\tilde{\pazocal{D}}=\{\tilde{y}_{[1,N]}^{e,d},\tilde{u}_{[1,N]}^{e,d}\}$ satisfying the persistence of excitation condition (see \cite[Theorem 17]{markovsky2022identifiability})  
       \begin{equation}\label{eq:generalized_condition}
        \mathrm{rank}(\pazocal{H}_{L}(\tilde{w}_{[1,N]}^{e,d}))=n_x+S(n_u+1)L,
    \end{equation}
     then for any $\tilde{w}_{[1,L]}^{e} \in \tilde{\mathcal{B}}_{|L}$ there exists $\tilde{g} \in \mathbb{R}^{N-L+1}$ such that
    \begin{equation}\label{eq:extended_data_driven}
        \pazocal{H}_{L}(\tilde{w}_{[1,N]}^{e,d})\tilde{g}=\tilde{w}_{[0,L-1]}^{e}.
    \end{equation}
    In turn, provided a feasible mode sequence $s_{[0,L-1]} \in \mathcal{B}_{[S]|L}^{\mathrm{PWA}}$, such a data-driven characterization can be \emph{restricted} to represent $\mathcal{B}_{s|L}^{\mathrm{PWA}}$ in \eqref{eq:fixed_s_behavior_truncated}. Specifically, by explicitly reintroducing the dependence of $\tilde{w}_{[0,L-1]}^{e}$ on $s_{[0,L-1]}$ as
    \begin{equation}
        \tilde{w}_{[0,L-1]}^{e}\!=\!\!\left[\mathrm{bdiag}(\mathds{1}(s_{[0,L-1]}))\otimes I_{n_y+n_u+1}\right]w_{[0,L-1]},
    \end{equation}
    where (with an abuse of notation) $\mathds{1}(s_{[0,L-1]})$ indicates the column vector stacking $\mathds{1}(s_{k})$ for $k \in [0,L-1]$, we enforce a specific sparsity pattern on the right-hand side of \eqref{eq:extended_data_driven}, which has to be matched by its left-hand side. By relying on \eqref{eq:constraints_for_equality}, this matching can be imposed by further restricting \eqref{eq:extended_data_driven} to
    \begin{equation}
        \pazocal{H}_{L}^{s}(\tilde{w}_{[1,N]}^{e,d})\tilde{g}\!=\!\!\left[\mathrm{bdiag}(\mathds{1}(s_{[0,L-1]}))\otimes I_{n_y+n_u+1}\right]w_{[0,L-1]},
    \end{equation}
    where $$\pazocal{H}_{L}^{s}(\tilde{w}_{[1,N]}^{e,d})=\left[\pazocal{S}_{n_y+n_u+1}\boldsymbol{1}_{N-L+1}^{\top}\right] \odot \pazocal{H}_{L}(\tilde{w}_{[1,N]}^{e,d}),$$
    and $\pazocal{S}_{n_y+n_u+1}\!=\!\mathds{1}(s_{[0,L-1]})\otimes \boldsymbol{1}_{n_y+n_u+1}$.
   This relationship features a set of trivial equalities (namely, $\mathbf{0}\!=\!\mathbf{0}$), which can be removed to obtain the data-driven characterization of $\mathcal{B}_{s|L}^{\mathrm{PWA}}$ matching the dimension in \eqref{eq:dimension_behavior}. 

   While this approach allows us to find a data-driven representation of $\mathcal{B}_{s|L}^{\mathrm{PWA}}$, it disregards the \emph{local affine nature} of the PWA system highlighted by \eqref{eq:kernel_representation}. Indeed, the latter shows that the kernel representation of a PWA system is a weighted sum of $S$ affine kernels, associated to the $S$ modes of the system, providing a different (and explainable) perspective on $\mathcal{B}_{s|L}^{\mathrm {PWA}}$ than the LTI embedding. Therefore, rather than exploiting the latter, we will leverage this insight, showing $(i)$ how to characterize affine local behaviors from data, and introducing $(ii)$ a data-driven characterization of PWA behaviors. 

    \subsection{Representing local behaviors from data}\label{sec:local_models}
    Let us consider the following \textquotedblleft piece\textquotedblright \ of \eqref{eq:kernel_representation}:
    \begin{equation}\label{eq:local_kernel}
        \sum_{j=0}^{\rho}r_j(i)q^{-j}\tilde{w}_t^{i}=\mathbf{0},
    \end{equation}
    where $\tilde{w}_{t}^{i}=\mathrm{col}(\tilde{y}_{t}^{i},\tilde{u}_{t}^{i},\mathds{1}(s_t=i))$ and 
        \begin{equation}\label{eq:useful_relations_tilde}
            \tilde{y}_{t-j}^{i}=\mathds{1}(s_t=i)y_{t-j},~~~\tilde{u}_{t-j}^{i}=\mathds{1}(s_t=i)u_{t-j},
        \end{equation}
        are extended signals embedding the mode selector, for all $j \!\in\! [0,\rho]$ and for all $t \!\in\! \mathbb{Z}_{\geq 0}$. According to this kernel representation, the \emph{affine behavior} of the $i$-th mode of the PWA system \eqref{eq:PWARX_sys} is dictated by
	\begin{equation}\label{eq:local_behavior}
		\mathcal{B}_{i}^{\mathrm{A}}=\left\{\tilde{w}^{i}\in \mathbb{W}^{\mathbb{Z}_{\geq 0}} \big| R_i(q)\tilde{w}^{i}=\mathbf{0}\right\},
	\end{equation}
	where $\mathbb{W}\!=\!\mathbb{R}^{n_u+n_y+1}$ and $R_i(q)\!=\!\sum_{j=0}^{\rho} r_{j}(i)q^{-j}$. This guarantees that the local affine behavior $\mathcal{B}_{i}^{\mathrm{A}}$ is shift-invariant and complete, and thus that $\mathcal{B}_{i}^{\mathrm{A}} \!\subseteq\! \mathcal{B}^{\mathrm{A}}$, where $\mathcal{B}^{\mathrm{A}}$ characterizes the class of affine systems as defined below\footnote{An alternative definition of affine systems can be found in \cite{markovsky2025represent}.}.
	\begin{definition}[Affine systems]\label{def:affine_systems}
		The tuple $(\mathbb{Z}_{\geq 0},\mathbb{W},\mathcal{B}^{\mathrm{A}})$ represents the class of affine systems if $\mathbb{W} \subseteq \mathbb{R}^{n_u+n_y+1}$ and $\mathcal{B}^{\mathrm{A}}$ is a shift-invariant, complete subspace of $\mathbb{W}^{\mathbb{Z}_{\geq 0}}$, i.e., closed in the topology of point-wise convergence (see~\cite[Section VI]{willems1986time}).
	\end{definition}
	From $\mathcal{B}_{i}^{\mathrm{A}}$, we can then define the truncated behavior of the $i$-th mode over a horizon $L$ as
	\begin{equation}\label{eq:truncated_local_behavior}
		\mathcal{B}_{i|L}^{\mathrm{A}}\!\!=\!\left\{\tilde{w}^{i\!} \!\in\! \mathbb{W}^{L} \big|  \exists \bar{w}^{i\!} \!\in\! \mathcal{B}_{i}^{\mathrm{A}\!} \mbox{ s.t. } \tilde{w}_{t}^{i}\!=\!\bar{w}_{t}^{i}, \forall  t \!\in\! [0,L\!-\!1]\right\}\!,\!\!
	\end{equation}   
	which can be characterized from data using the following extension of Willems' fundamental lemma~\cite{martinelli2022data}.
	\begin{theorem}[Mode's data-driven representation]\label{thm:affine_dd}
		Let the $i$-th affine sub-system with behavior $\mathcal{B}_{i}^{\mathrm{A}}$ in \eqref{eq:local_behavior} be controllable. Moreover, let $y_{[1,N_{i}]}^{i,d}$ be the outputs generated by it in response to an input sequence $u_{[1,N_{i}]}^{i,d}$ that is persistently exciting of order $n_x+L+1$. Then, $\{{u}_{[1,L-1]}^{i},{y}_{[1,L-1]}^{i}\}$ is an input-output trajectory of the $i$-th mode of \eqref{eq:PWARX_sys} if and only if there exists $\pazocal{G}_{i} \in \mathbb{R}^{N_{i}-L+1}$ such that 
		\begin{subequations}\label{eq:data_driven_local_behavior}
			\begin{align}
				& \begin{bmatrix}\pazocal{H}_{L}(u_{[1,N_{i}]}^{i,d})\\
					\pazocal{H}_{L}(y_{[1,N_{i}]}^{i,d})
				\end{bmatrix}\pazocal{G}_{i}=\begin{bmatrix}
					u_{[0,L-1]}^{i}\\
					y_{[0,L-1]}^{i}\end{bmatrix},\\
				& \boldsymbol{1}_{N_i-L+1}^{\top}\pazocal{G}_{i}=1,
			\end{align}
		\end{subequations}
		where $\pazocal{H}_{L}(u_{[1,N_{i}]}^{i,d})$ and $\pazocal{H}_{L}^{i}(y_{[1,N_{i}]}^{i,d})$ are Hankel data matrices defined as in \eqref{eq:Hankel}.
	\end{theorem}
	\begin{proof}
		The proof follows from that of \cite[Theorem 1]{martinelli2022data} and it is thus omitted.
	\end{proof}
	\begin{remark}[Complexity \& local behaviors]
		The quantity $\rho\!=\!\max\{n_a,n_b\}$ in \eqref{eq:local_kernel} and \eqref{eq:kernel_representation} represents the lag of each mode of the PWA system, and it is common to all of them. Together with the lag, the complexity of local representations depends on the dimension of the state $n_x$ characterizing the (minimal) state-space representation in \eqref{eq:PWA_sys}. 
        \hfill $\square$ 
	\end{remark}
	\begin{remark}
		Since \eqref{eq:PWARX_sys} is affine along the mode sequence, \eqref{eq:fixed_s_behavior} is \emph{affine}. Namely, each element of $w \!\in\! \mathcal{B}^{\mathrm{PWA}}_{s}$ satisfies $w_t\! \in\! \mathcal{B}_{s_t}^{\mathrm{A}}$, with $\mathcal{B}_{s_t}^{\mathrm{A}}$ defined as in \eqref{eq:local_behavior}, for all $t \!\in\! \mathbb{Z}_{\geq 0}$. \hfill $\square$
	\end{remark}
	\subsection{PWA behaviors from data with structural insights}\label{sec:PWA_behavior}
    Let us now reconsider the kernel representation in \eqref{eq:kernel_representation}, and equivalently rewrite it based on \eqref{eq:local_kernel} as
    \begin{equation}\label{eq:kernel_equivalent}
        \sum_{i=1}^{S}\sum_{j=0}^{\rho} r_j(i)q^{-j}\tilde{w}_t^{i}=\mathbf{0}.
    \end{equation}
    This equivalent representation shows the PWA behavior as the sum of the local behaviors of its $S$ modes. To this end, let us make the following assumptions on the data $\pazocal{D}$ in \eqref{eq:data}.
    \begin{assumption}[Known data-generating active mode]\label{assumption:mode_access}
    The 
    active mode sequence $s_{[1,N]}^{d}$ associated with the input/output pairs in the dataset $\pazocal{D}$ in \eqref{eq:data} is known
    . \hfill $\square$ 
    \end{assumption}    
    Under the previous assumption, we partition $\pazocal{D}$ into $S$ subsets
	\begin{equation}\label{eq:data_local}
		\pazocal{D}_i=\left\{u_{[1,N_i]}^{i,d},y_{[1,N_i]}^{i,d}\right\},~~i\in [S],
	\end{equation} 
    such that $\pazocal{D}=\bigcup_{i=1}^{S}\pazocal{D}_i$, with $N=\sum_{i=1}^{S}N_i$. Let us further assume that the local datasets satisfy the following. 	
    \begin{assumption}[Subset characteristics]\label{assumption:dataset_features}
	Each subset $\pazocal{D}_{i}$ in \eqref{eq:data_local} is collected by exciting the $i$-th mode of the PWA system with an input $u_{[1,N_i]}^{i,d}$ that is persistently exciting of order $n_x+L+1$, for all $i \in [S]$. \hfill $\square$ 
	\end{assumption}
	While the former assumption will be relaxed in the sequel of the paper, the latter guarantees that we can construct trajectories of length $L$ for all modes of the PWA system, as well as reconstruct the associated partition of the input/output space satisfying \eqref{eq:PWARX_partition}. At the same time, it implies that the length of the overall dataset is $N=\sum_{i=1}^{S}N_i$, with $N_i$ large enough to satisfy Assumption \ref{assumption:dataset_features}. 
    
    By rearranging terms, these assumptions allow us to characterize the (extended) behavior of each mode from the data using \eqref{eq:data_driven_local_behavior}. Exploiting such a representation and the result in \cite[Proof of Theorem 1]{fazzi2023addition}, we obtain 
    \begin{equation}\label{eq:tilde_data_driven_behavior2}
		\sum_{i=1}^{S}\begin{bmatrix}
			{\pazocal{H}}_{L}(u_{[1,N_{i}]}^{i,d})\\
			{\pazocal{H}}_{L}(y_{[1,N_{i}]}^{i,d})
		\end{bmatrix}\pazocal{G}_i=\begin{bmatrix}
			\tilde{u}_{[0,L-1]}\\
			\tilde{y}_{[0,L-1]}
		\end{bmatrix},
    \end{equation}
    where
    \begin{equation}\label{eq:utils_datarep}
        \begin{aligned}
            \tilde{u}_t\!=\!\!\sum_{i=1}^{S}\tilde{u}_t^i,~~~\tilde{y}_t\!=\!\!\sum_{i=1}^{S}\tilde{y}_t^i,
        \end{aligned}
    \end{equation}
    as in \eqref{eq:useful_relations_tilde}, $\pazocal{G}_{i} \in \mathbb{R}^{N_i-L+1}$ still has to satisfy    \begin{equation}\label{eq:local_affine_equalities}
        \boldsymbol{1}_{N_i-L+1}^{\top}\pazocal{G}_{i}=1, ~\forall i\in S,
    \end{equation}
    for local behaviors to be affine, while $\pazocal{H}_{L}(u_{[1,N_i]}^{i,d})$ and $\pazocal{H}_{L}(y_{[1,N_i]}^{i,d})$ are $Ln_{u}\times (N_i-L+1)$ and $Ln_{y}\times (N_i-L+1)$ dimensional matrices containing local data only. Note that the following equivalence holds:
    \begin{equation}\label{eq:Mosaic_modes}
			\sum_{i=1}^{S}\begin{bmatrix}
			{\pazocal{H}}_{L}(u_{[1,N_{i}]}^{i,d})\\
			{\pazocal{H}}_{L}(y_{[1,N_{i}]}^{i,d})
		\end{bmatrix}\pazocal{G}_i=\begin{bmatrix}\pazocal{M}_L(u_{[1,N]}^{d})\\
			\pazocal{M}_L(y_{[1,N]}^{d})
		\end{bmatrix}g,
    \end{equation}
    with 
    \begin{equation}\label{eq:Mosaic_modes_dec}
		\begin{bmatrix}
			\pazocal{M}_L(u_{[1,N]}^{d})\\
			\pazocal{M}_L(y_{[1,N]}^{d})
		\end{bmatrix}\!\!=\!\!\begin{bmatrix}
			\pazocal{H}_{L}(u_{[1,N_{1}]}^{i,d}) \!\!&\!\! \cdots \!\!&\!\! \pazocal{H}_{L}(u_{[1,N_{S}]}^{S,d})\\
			\pazocal{H}_{L}(y_{[1,N_{1}]}^{1,d}) \!\!&\!\! \cdots \!\!&\! \pazocal{H}_{L}(y_{[1,N_{S}]}^{S,d})
		\end{bmatrix}\!,
    \end{equation}
    and
    \begin{equation}\label{eq:g_decomposition}
		g=\begin{bmatrix}
			\pazocal{G}_{1}^{\top} &\pazocal{G}_{2}^{\top} & \cdots & \pazocal{G}_{S}^{\top}
		\end{bmatrix}^{\!\top} \in \mathbb{R}^{N-SL+S},
    \end{equation}
    thus showing that the adopted data structure aligns with those used in subspace identification~\cite{verdult2004subspace}. At the same time, what we have obtained is not yet a data-driven representation of $\mathcal{B}^{\mathrm{PWA}}$, as the dependence on the mode sequence is missing. 

    Let $u_{[0,L-1]}$, $y_{[0,L-1]}$ be the actual input/output trajectories of the PWA associated with the feasible mode sequence $s_{[0,L-1]}\!\in\! \mathcal{B}_{[S]|L}^{\mathrm{PWA}}$. Due to \eqref{eq:useful_relations_tilde}, the following holds  
    \begin{equation}
        \begin{bmatrix}
			\tilde{u}_{[0,L-1]}\\
			\tilde{y}_{[0,L-1]}
		\end{bmatrix}\!=\!\sum_{i=1}^{S}\! \underbrace{\begin{bmatrix}
		    \mathds{1}(s_{[0,L-1]}=i)\!\otimes\! \boldsymbol{1}_{n_u}\\
            \mathds{1}(s_{[0,L-1]}=i)\!\otimes \! \boldsymbol{1}_{n_y}
		\end{bmatrix}}_{:=\pazocal{S}^{i}} \odot \begin{bmatrix}
			u_{[0,L-1]}^{i}\\
			y_{[0,L-1]}^{i}
		\end{bmatrix}\!,
    \end{equation}
     where $\mathds{1}(s_{[0,L-1]}\!=\!i) \!\in\! \{0,1\}^{L}$ indicates (with a slight abuse of notation) the vector stacking zeros and ones depending on whether the $i$-th mode is inactive or active over the horizon $L$, for all $i \!\in\! [S]$. Applying the same rationale to the sum on the left-hand side of \eqref{eq:tilde_data_driven_behavior2}, we obtain a restricted data-driven representation
     \begin{equation}\label{eq:data_driven_behavior}
			\sum_{i=1}^{S}
				\tilde{\pazocal{S}}^{i\!\!}
				\odot\!\! 
			\begin{bmatrix}\pazocal{H}_{L}(u_{[1,N_{i}]}^{i,d})\\
				\pazocal{H}_{L}(y_{[1,N_{i}]}^{i,d})
			\end{bmatrix}\!\pazocal{G}_{i}\!=\!\!\!\sum_{i=1}^{S} \pazocal{S}^{i\!} \odot\!\! \begin{bmatrix}
			u_{[0,L-1]}^{i}\\
			y_{[0,L-1]}^{i}
		\end{bmatrix}\!\!=\!\! \begin{bmatrix}
		    u_{[0,L-1]}\\
                y_{[0,L-1]}
		\end{bmatrix}\!,
    \end{equation}
    where 
    \begin{equation}\label{eq:S_tilde}
      \tilde{\pazocal{S}}^{i}=\pazocal{S}^{i}\boldsymbol{1}_{N_i-L+1}^\top.   
    \end{equation}
    By additionally imposing \eqref{eq:local_affine_equalities}, the previous equality  
    characterizes a trajectory $\mathrm{col}(u_{[0,L-1]},y_{[0,L-1]}) \in \mathcal{B}_{s|L}^{\mathrm{PWA}}$, leading to a data-driven representation of the behavior of the PWA system \eqref{eq:PWARX_sys} compatible with a given mode sequence $s_{[0,L-1]}$. 
    
    We now note that, according to \eqref{eq:PWARX_partition1}-\eqref{eq:PWARX_partition2}, $s_{[0,L-1]}$ in turn depends on $w_{[0,L-1]}$ and, hence on $\{\pazocal{G}_{i}\}_{i=1}^{S}$ (see \eqref{eq:data_driven_behavior}). Indeed, by noticing that the following holds:
    \begin{equation}\label{eq:selected_partial_traj}
        \xi_{[t-\rho,t-1]}=\sum_{i=1}^{S}\underbrace{\pazocal{S}_{\xi,[t-\rho,t-1]}^{i}\odot \pazocal{H}_{\rho}^{t}(\xi_{[1,N_i]}^{i,d})}_{:=\pazocal{H}_{\rho}^{\pazocal{S},t}(\xi_{[1,N_i]}^{i,d})}\pazocal{G}_{i},
    \end{equation}    
   where \begin{align*}
       & \pazocal{S}_{\xi,[t-\rho,t-1]}^{i}=(\mathds{1}(s_{[t-\rho,t-1]}=i)\otimes \mathbf{1}_{n_{\xi}})\mathbf{1}_{N_i-L+1}^{\top},\\
       & \pazocal{H}_{\rho}^{t}(\xi_{[1,N_i]}^{i,d})=[\pazocal{H}_{L}(\xi_{[1,N_i]}^{i,d})]_{n_{\xi}(t-\rho)+1:n_{\xi}(t-1)},
   \end{align*} 
   and $\xi$ is a placeholder for $y$ and $u$, we get that
    \begin{equation}\label{eq:G_dependent_assignment}
        s_t=i \iff P_i\begin{bmatrix}
            \sum_{j=1}^{S}
            \pazocal{H}_{\rho}^{\pazocal{S},t}(y_{[1,N_j]}^{j,d})\pazocal{G}_{j}
            \\
                        \sum_{j=1}^{S}
           \pazocal{H}_{\rho}^{\pazocal{S},t}(u_{[1,N_j]}^{j,d})\pazocal{G}_{j}
            \\
            u_t\\
            1
        \end{bmatrix}\preceq \mathbf{0},
    \end{equation}
    with $P_i$ defined in \eqref{eq:PWARX_partition2}, for $i \in [S]$. Although the previous relationship is well defined for $t \in [\rho,L-1]$, it will depend on the inputs and outputs at negative time moments when $t \in [0,\rho-1]$. In this case, we rely on the following assumption.
    \begin{assumption}[On past inputs/outputs]\label{assumption:past_data}
        Inputs and outputs satisfy $y_{t}=\mathbf{0}$ and $u_t=\mathbf{0}$ for all $t<0$, respectively. \hfill $\square$
    \end{assumption}
    Note that in characterizing the active mode at time $t$, we do not replace the input $u_t$ with its data-driven representation since the latter depends on $s_t$ (i.e., the mode to be characterized). This nested dependence is present only if the system is not strictly proper.     \begin{remark}[Connection with PWA identification]
        Based on our computation, it is clear that jointly characterizing $\mathcal{B}_{s|L}^{\mathrm{PWA}}$ and $\mathcal{B}_{[S]|L}^{\mathrm{PWA}}$ from data is more demanding than solely characterizing one of the two elements assuming the other to be fixed. This finding is in line with a common practice in PWA identification approaches, which often alternate (iteratively) between estimating one unknown, either the local models or the active mode sequence, while fixing the other (see, e.g.,~\cite{yu2025,bemporad2003greedy,ferraritrecate2003,Bemporad2005,Breschi2016}). \hfill $\square$
    \end{remark}
    
    By using \eqref{eq:data_driven_behavior} and \eqref{eq:G_dependent_assignment}, we have now characterized $\mathcal{B}^{\mathrm{PWA}}_{|L}$ from data. This result allows us to state the following extension of Willems' fundamental lemma~\cite{willems2005note} for PWA systems.
    \begin{theorem}[PWA Fundamental lemma]\label{thm:fundamental_PWA}
    Consider a controllable PWA system described by \eqref{eq:PWARX_sys}, whose modes are all controllable. Let $\pazocal{D}$ be a dataset generated by such a system satisfying Assumption~\ref{assumption:dataset_features}. Under Assumptions~\ref{assumption:mode_access} and~\ref{assumption:past_data}, for a given $L\geq\rho$, $\left\{u_{[0,L-1]},y_{[0,L-1]}\right\}$ is an input-output trajectory and $s_{[0,L-1]}$ is a feasible mode sequence of \eqref{eq:PWARX_sys}, i.e., $\left\{u_{[0,L-1]},y_{[0,L-1] },s_{[0,L-1]}\right\} \!\in\! \mathcal{B}_{|L}^{\mathrm{PWA}}$, if and only if there exists $g \in \mathbb{R}^{N-SL+S}$ in \eqref{eq:g_decomposition} satisfying \eqref{eq:local_affine_equalities}, \eqref{eq:data_driven_behavior} and \eqref{eq:G_dependent_assignment}.
    \end{theorem}
    \begin{proof}
	   The proof can be found in Appendix~\ref{appendix:fundamental_lemma}.
    \end{proof} 

\section{Does shrinkage-based DeePC lead to actions coherent with a PWA system?}\label{sec:shrinkage_based}

We now look into the design of a data-enabled predictive control (DeePC) scheme for the PWA system in \eqref{eq:PWA_sys} to track a reference input/output trajectory $(u_t^{\mathrm{o}},y_{t}^{\mathrm{o}})$, of which we assume to have a full preview for $t \in \mathbb{Z}_{\geq0}$. Nonetheless, in this work, we do not use \eqref{eq:data_driven_behavior} to construct a data-driven predictor, which would require solving a mixed-integer problem to concurrently assign the active mode and optimize the local selectors $\{\pazocal{G}_i\}_{i=1}^{S}$. Instead, we want to analyze whether leveraging the \emph{structural} insights provided by \eqref{eq:data_driven_behavior} and using two specific \emph{shrinking schemes} is enough for DeePC to make decisions that are \emph{coherent} with the PWA nature of the controlled system. Coherence is here (loosely) intended as the capability of DeePC to leverage only data associated with the mode that should be active to take decisions at each step in the prediction horizon (see Definition~\ref{def:coherence} for a more formal statement).

To this end, let us consider the local datasets $\pazocal{D}_{i}$ constructed still under Assumption~\ref{assumption:mode_access} and use them to fill the local Hankel matrices $\pazocal{H}_{L+\rho}(u_{[1,N_{i}]}^{i,d})$ and $\pazocal{H}_{L+\rho}(y_{[1,N_{i}]}^{i,d})$, for all $i\! \in\! [S]$. Let us then slice these matrices as
\begin{equation}\label{eq:sliced_matrices}
	\pazocal{H}_{L+\rho}(u_{[1,N_{i}]}^{i,d})=\begin{bmatrix}
		U_{P}^{i}\\
		U_{F}^{i}
	\end{bmatrix},~~~\pazocal{H}_{L+\rho}(y_{[1,N_{i}]}^{i,d})=\begin{bmatrix}
		Y_{P}^{i}\\
		Y_{F}^{i}
	\end{bmatrix}\!,
\end{equation}
with 
\begin{align*}
&U_{P}^{i} \in \mathbb{R}^{n_u\rho\times N_i-L-\rho-1},~~~U_{F}^{i} \in \mathbb{R}^{n_uL\times N_i-L-\rho-1},\\
&Y_{P}^{i} \in \mathbb{R}^{n_y\rho\times N_i-L-\rho-1},~~~Y_{F}^{i} \in \mathbb{R}^{n_yL\times N_i-L-\rho-1},
\end{align*}
respectively. These matrices to form the DeePC problem 
\begin{subequations}\label{eq:general_deepc}
	\begin{align}
		& \underset{\substack{u_{[0,L-1]},~y_{[0,L-1]}\\\{\pazocal{G}_{i}\}_{i=1}^{S}}}{\mathrm{minimize}}J(u_{[0,L-1]},y_{[0,L-1]})\!+\!\! \sum_{j=1}^{M}\lambda_{j}\!\!\sum_{i=1}^{S}r_i^{j}(\pazocal{G}_i)\label{eq:general_deepc_reg}\\
		&\qquad \qquad \mbox{s.t. }\quad \sum_{i=1}^{S}\begin{bmatrix}
			Z_P^i\\
			U_{F}^{i}\\
			Y_{F}^{i}
		\end{bmatrix}\pazocal{G}_i=\begin{bmatrix}
			z_{\mathrm{ini}}\\
			u_{[0,L-1]}\\
			y_{[0,L-1]}
		\end{bmatrix},\label{eq:DeePC_model}\\
		& \qquad \qquad \qquad ~~ \boldsymbol{1}_{N_i-L-\rho+1}^{\top}\pazocal{G}_i=1,~~~\forall i \in [S],\label{eq:DeePC_affine_constr}\\
		&\qquad \qquad \qquad ~~ u_k \in \pazocal{U},~~~\forall k \in [0,L-1],\label{eq:DeePC_input_constr}\\
		&\qquad \qquad \qquad ~~ y_k \in \pazocal{Y},~~~\forall k \in [0,L-1],\label{eq:DeePC_output_constr}
	\end{align}
	to be solved at each time instant $t \in \mathbb{Z}_{\geq 0}$ in a \emph{receding horizon} fashion~\cite[Section 12.1]{borrelli2017predictive}, with $Z_P^i=\smallmat{(U_P^i)^{\top} & (Y_P^i)^{\top}}^{\top}$, and 
	\begin{equation}
		z_{\mathrm{ini}}=\begin{bmatrix}
			u_{[t-\rho,t-1]}^{\top} & y_{[t-\rho,t-1]}^{\top} 
		\end{bmatrix}^{\top}\!\in \mathbb{R}^{(n_u+n_y)\rho},
	\end{equation}
	used as a proxy for the current state of the system. 
	This problem features input and output constraints, with both $\pazocal{U} \!\subseteq\! \mathbb{R}^{n_u}$ and $\pazocal{Y} \!\subseteq\! \mathbb{R}^{n_y}$ assumed to be \emph{polytopic}. Meanwhile, its cost comprises both the conventional quadratic tracking loss
	\begin{equation}\label{eq:DeePCloss}
		J(u_{[0,L-1]},y_{[0,L-1]})\!=\!\!\sum_{k=0}^{L-1}\!\left[\|y_{k}\!-\!y_{k}^{\mathrm{o}}\|_{Q}^{2\!}\!+\!\|u_{k}\!-\!u_{k}^{\mathrm{o}}\|_{R}^{2}\right]\!,
	\end{equation}
	with $Q\! \in\! \mathbb{R}^{n_y \times n_y}$, $Q\! \succeq \!\mathbf{0}$, and $R \in \mathbb{R}^{n_u \times n_u}$ with $R \!\succ \!\mathbf{0}$,	and a linear combination of $M\!\geq\! 1$ regularization functions $r_i^{j}:\mathbb{R}^{N_i-L-\rho+1} \!\rightarrow\! \mathbb{R}^{+}_{0}$ for all $i \!\in\! [S]$ driven by the regularization penalties $\lambda_{j}\!>\!0$, with $j\! \in\! [M]$. These regularization terms will be selected afterwards to enforce sparsity into $\{\pazocal{G}_i\}_{i=1}^{S}$.
\end{subequations}

Irrespective of the chosen regularization term, it is always possible to recast \eqref{eq:general_deepc} as an optimization problem solely on $\{\pazocal{G}_i\}_{i=1}^{S}$. To this end, let us define $\pazocal{Q}=\mathrm{diag}(Q,\ldots,Q)$ and $\pazocal{R}=\mathrm{diag}(R,\ldots,R)$, with $\pazocal{Q} \in \mathbb{R}^{n_yL\times n_yL}$ and $\pazocal{R} \in \mathbb{R}^{n_uL\times n_uL}$. By leveraging \eqref{eq:DeePC_model}, the loss \eqref{eq:DeePCloss} can be equivalently recast as
	\begin{equation}\label{eq:DeePCloss2}
		J(\{\pazocal{G}_i\}_{i=1}^{S})\!= \bigg\| \sum_{i=1}^{S}Y_F^{i}\pazocal{G}_i-\boldsymbol{y^{\mathrm{o}}} \bigg\|_{\pazocal{Q}}^{2}+\bigg\| \sum_{i=1}^{S}U_F^{i}\pazocal{G}_i-\boldsymbol{u^{\mathrm{o}}} \bigg\|_{\pazocal{R}}^{2},
	\end{equation}
where $\boldsymbol{y^{\mathrm{o}}} \!\in\! \mathbb{R}^{n_yL}$ and $\boldsymbol{u^{\mathrm{o}}} \!\in\! \mathbb{R}^{n_uL}$ are compactly denoting $y_{[t,t+L-1]}^{\mathrm{o}}$ and $u_{[t,t+L-1]}^{\mathrm{o}}$, respectively. Meanwhile, thanks to the polytopic nature of \eqref{eq:DeePC_input_constr}-\eqref{eq:DeePC_output_constr}, the input/output constraints can be equivalently summarized as
\begin{equation}\label{eq:constraints2}
	\sum_{i=1}^{S}\Gamma_{i}\pazocal{G}_i \leq \gamma,
\end{equation} 
where $\Gamma_{i}$ depends on $U_F^{i}$ and $Y_{F}^{i}$, for all $i \in [S]$, as well as the features of the input/output polytopic constraints, while $\gamma$ opportunely stacks the latter's affine terms. Moreover, let us equivalently rewrite the equality constraint in \eqref{eq:DeePC_affine_constr} as
\begin{equation}
    \sum_{i=1}^{S} \pazocal{I}_i\pazocal{G}_i=\boldsymbol{1}_{S},
\end{equation}
where $\pazocal{I}_i$ denotes the $i$-th sub-block of the matrix $\mathrm{bdiag}(\boldsymbol{1}_{N_1-L+1}^{\top},\ldots,\boldsymbol{1}_{N_S-L+1}^{\top})$. We can then define 
\begin{equation}
	\tilde{Z}_{P}^{i}=\begin{bmatrix}
	Z_{P}^{i}\\
	\pazocal{I}_i
	\end{bmatrix},~~~\tilde{z}_{\mathrm{ini}}=\begin{bmatrix}
	z_{\mathrm{ini}}\\
	\boldsymbol{1}_{S}
	\end{bmatrix},
\end{equation}
so that \eqref{eq:general_deepc} can be equivalently recast as
\begin{subequations}\label{eq:general_deepc2}
	\begin{align}
		& \underset{\{\pazocal{G}_{i}\}_{i=1}^{S}}{\mathrm{minimize}}~~J(\{\pazocal{G}_{i}\}_{i=1}^{S})+ \sum_{j=1}^{M}\lambda_{j}\sum_{i=1}^{S}r_i^{j}(\pazocal{G}_i)\\
		&\qquad ~ \mbox{s.t. }~ \sum_{i=1}^{S}\tilde{Z}_{P}^{i}
		\pazocal{G}_i=\tilde{z}_{\mathrm{ini}},\label{eq:DeePC_init_constr}\\
		&\qquad \qquad  \sum_{i=1}^{S}\Gamma_i \pazocal{G}_i \leq \gamma.\label{eq:DeePC_ineq_constr}
	\end{align}
\end{subequations}
According to \eqref{eq:general_deepc2}, we can now 
 draw from \cite[Definition~3]{giacomelli2025insights} and formalize our notion of coherence of DeePC's decisions with respect to the PWA nature of the controlled system. 
 \begin{definition}[DeePC coherence with PWA behavior]\label{def:coherence}
    The solution $\{\pazocal{G}_i^\star\}_{i=1}^{S}$ of the DeePC problem in \eqref{eq:general_deepc2} is coherent with the PWA behavior if $\pazocal{G}_{j}^\star\!=\!\mathbf{0}$ for all $j \!\neq\! i$, $i,j \!\in\! [S]$, whenever the initial condition $z_{\mathrm{ini}}$ and the references $(\boldsymbol{u^{\mathrm{o}}},\boldsymbol{y^{\mathrm{o}}})$ demand only the $i$-th mode of \eqref{eq:PWARX_sys} to be active. \hfill $\square$
 \end{definition}
Note that, if multiple modes are required to be active, then \eqref{eq:general_deepc2} should be able to at least prioritize data associated with the \textquotedblleft right\textquotedblright \ operating mode at each prediction step. 

\begin{remark}[An alternative for coherence analysis] 
    While our strategy to analyze the coherence of DeePC schemes relies on deriving their explicit solutions, one could alternatively adapt the implicit predictor approach based on the  (see~\cite{kladtke2023implicit}) as done for the linear framework in \cite{kladtke2025data,kladtke2025towards}.  
\end{remark}
 
\subsection{Elastic-DeePC \& its coherence with PWA behaviors}\label{sec:Elastic_DeePC}
Similarly to \cite[Section III]{giacomelli2025insights}, we first consider an \emph{elastic net} regularization on $\{\pazocal{G}_i\}_{i=1}^{S}$, i.e.,
\begin{equation}\label{eq:elastic_reg}
	\sum_{j=1}^{M} \lambda_{j}\sum_{i=1}^{S}r_{i}^{j}(\pazocal{G}_i)=\lambda_{1}\sum_{i=1}^{S}\|\pazocal{G}_{i}\|_{1}+\lambda_{2}\sum_{i=1}^{S}\|\pazocal{G}_{i}\|_{2}^{2},
\end{equation}
and find an explicit solution for \eqref{eq:general_deepc2} when using it. To this end, we now introduce the augmented Lagrangian associated with the Elastic-DeePC problem, namely
\begin{align}\label{eq:Lagrangian_elastic}
	\nonumber &\pazocal{L}(\{\pazocal{G}_{i}\}_{i=1}^{S},\alpha,\mu;\lambda_{1},\lambda_{2})\!=\!J(\{\pazocal{G}_i\}_{i=1}^{S})\!+\!\lambda_{1}\sum_{i=1}^{S}\|\pazocal{G}_i\|_{1}\!\!\\
	& \!+\!\!\lambda_{2}\!\sum_{i=1}^{S}\!\|\pazocal{G}_i\|_{2}^{2\!}\!+\!\alpha^{\!\top\!\!}\left(\!
    \sum_{i=1}^{S}\tilde{Z}_{P}^{i}\pazocal{G}_i\!-\!\tilde{z}_{\mathrm{ini\!\!}}\!\right)\!\!+\!\mu^{\!\top\!\!\!}\left(\sum_{i=1}^{S}\Gamma_i\pazocal{G}_i\!-\!\gamma\!\!\right)\!\!,\!\!
\end{align}
where $\alpha$ and $\mu$ are vectors stacking the Lagrange multipliers associated with the equality and inequality constraints featured in \eqref{eq:general_deepc2}, respectively. Let us then assume only a subset of the inequalities in \eqref{eq:DeePC_ineq_constr} to be \emph{active} over the prediction horizon and let $\tilde{\Gamma}_{i}$ and $\tilde{\gamma}$ be the coefficients associated with them. Accordingly, \eqref{eq:DeePC_init_constr} can be rewritten as 
\begin{equation}\label{eq:DeePC_init_constr_ext}
    \sum_{i=1}^{S}\tilde{Z}_{i}\pazocal{G}_{i}^{\star}-\underbrace{\begin{bmatrix}
    \tilde{z}_{\mathrm{ini}}\\
    \tilde{\gamma}
    \end{bmatrix}}_{=\tilde{b}}=\mathbf{0},
\end{equation}
with
\begin{equation}\label{eq:equality_utils}
    \tilde{Z}_{i}=\begin{bmatrix}
    (\tilde{Z}_{P}^{i})^{\top} & \tilde{\Gamma}_i^{\top} 
\end{bmatrix}^{\top}.
\end{equation}
Similar to \cite[Assumption 2]{Breschi2023}, let us assume that the following holds for
the \textquotedblleft expanded\textquotedblright \ equality constraints in \eqref{eq:DeePC_init_constr_ext}. 
\begin{assumption}[Equalities well posedness I]\label{assumption:technical1}
    The rows of $\tilde{Z}_i$ are linearly independent for all $i \in [S]$.
\end{assumption}

Supposing that Assumption~\ref{assumption:technical1} holds for all possible combinations of active constraints, we can use the Karush-Kuhn-Tucker (KKT) conditions associated with \eqref{eq:Lagrangian_elastic} to find the unique, explicit solution of \eqref{eq:general_deepc2} with the regularization in \eqref{eq:elastic_reg}, whose result is formalized in the following proposition.
\begin{proposition}[Explicit solution of elastic-DeePC]\label{prop:explicit_elastic}
   Let $\lambda_1,\lambda_2$ in \eqref{eq:elastic_reg} be positive real numbers and Assumption~\ref{assumption:technical1} hold for all possible combinations of active constraints of \eqref{eq:general_deepc2}. Then, the explicit solutions of \eqref{eq:general_deepc2} with the regularization in \eqref{eq:elastic_reg} is given by a PWA law in the $\tilde{z}_{\mathrm{ini}}$ domain of the kind: 
    \begin{equation}\label{eq:explicit_elastic}
        \pazocal{G}_{i}^{\star}=\begin{cases}
          F_{1}^{i}\tilde{z}_{\mathrm{ini}},~~~~~\mbox{if}~~~ E_{1}^{i}\tilde{z}_{\mathrm{ini}}\preceq\mathbf{0}, \\
          \vdots\\
          F_{n_{M}^i}^{i}\tilde{z}_{\mathrm{ini}},~~~\mbox{if}~~~E_{n_M^i}^{i}\tilde{z}_{\mathrm{ini}}\preceq\mathbf{0},
        \end{cases}
    \end{equation}
   where the coefficient matrices $E_{m} \!\in\! \mathbb{R}^{n_m\times [(n_u+n_y)\rho+S]}$ and $F_{m} \!\in\! \mathbb{R}^{(N_i-L-\rho+1)\times [(n_u+n_y)\rho+S]}$ may depend on data from \emph{all} modes of \eqref{eq:PWARX_sys}, for all $m \in [n_M^i]$ and all $i \in [S]$.   
   \hfill $\square$
\end{proposition}
\begin{proof}
    See Appendix~\ref{appendix:proof_prop_elastic}.
\end{proof}
Hence, as the optimal selectors of Elastic-DeePC may depend on data from other modes of the PWA system, the reconstructed behavior might be incoherent with the local dynamics (see Definition~\ref{def:coherence}). Note that, even if we consider PWA systems rather than linear ones, this result is also aligned with the one of \cite[Theorem 1]{kladtke2025data}, further sustaining that data forming the solution may not be aligned with the system's \textquotedblleft true\textquotedblright \ behavior.

\subsection{CAP-DeePC \& its coherence with PWA behaviors}\label{sec:CAP_DeePC}
Given the impossibility of obtaining coherent behaviors with Elastic-DeePC, we explore the suitability of the Composite Absolute Penalty (CAP) regularization~\cite{zhao2006grouped} to match a sparsity pattern of the PWA system's behavior in line with the tracking objective, without solving a mixed integer DeePC. Following the suggestions in~\cite[Section 2.3.1]{zhao2006grouped}, we focus on a specific case of CAP regularization, i.e., the Group Lasso (see~\cite{yuan2006model})
\begin{equation}\label{eq:GLasso}
	\sum_{i=1}^{M} \lambda_{j}\sum_{i=1}^{S}r_{i}^{j}(\pazocal{G}_i)=\lambda\sum_{i=1}^{S}\sqrt{N_i-L-\rho+1}\|\pazocal{G}_{i}\|_{2},
\end{equation}
as no prior information about the data, apart from their grouping, is available.
We note that, in this setting, it is not possible to find a closed-form expression of the optimal selectors $\pazocal{G}_{i}^{\star}$ as a function of the data, as our problem does not generally match the particular structures mentioned in \cite[Chapter 4.3.1]{hastie2015statistical}. We thus focus on unveiling the conditions under which a local selector is shrunk to zero. To this end, let us introduce the Lagrangian associated with \eqref{eq:general_deepc2} with regularization \eqref{eq:GLasso}, i.e.,
\begin{align}\label{eq:Lagrangian_GLasso}
	\nonumber &\pazocal{L}(\{\pazocal{G}_{i}\}_{i=1}^{S},\alpha,\mu;\lambda)=J(\{\pazocal{G}_i\}_{i=1}^{S})+\sum_{i=1}^{S}\lambda_i\|\pazocal{G}_i\|_{2}\\
	& \qquad~~ +\alpha^{\!\top\!\!}\left(\sum_{i=1}^{S}\tilde{Z}_{P}^{i}\pazocal{G}_i-\tilde{z}_{\mathrm{ini}}\right)\!+\!\mu^{\!\top\!\!}\left(\sum_{i=1}^{S}\Gamma_i\pazocal{G}_i-\gamma\right),
\end{align}
where, with an abuse of notation, $\lambda_i=\lambda\sqrt{N_i-L-\rho+1}$, for all $i \in [S]$.
Accordingly, the KKT conditions associated with the problem of interest are
\begin{subequations}\label{eq:KKT_GLasso}
\begin{align}
& W_{F}^{i}\pazocal{G}_{i}^{\star\!}+c_{\neq i\!}+\lambda_i\omega_{i}^{\star\!}+(\tilde{Z}_{P}^{i})^{\!\top\!\!}\alpha\!+\!(\Gamma_i)^{\!\top\!\!}\mu\!=\!\mathbf{0},~~\forall i \!\in\! [S],\label{eq:KKT_GLasso_1}\\
& \sum_{i=1}^{S}\tilde{Z}_{P}^{i}\pazocal{G}_{i}^{\star}=\tilde{z}_{\mathrm{ini}},\label{eq:KKT_GLasso_2}\\
& \sum_{i=1}^{S}\Gamma_{i}\pazocal{G}_{i}^{\star}\leq \gamma,\label{eq:KKT_GLasso_3}\\
& \sum_{i=1}^{S} \mu^{\top}\left(\Gamma_{i}\pazocal{G}_{i}^{\star}- \gamma\right)=0,\label{eq:KKT_GLasso_4}\\
& \mu \geq \mathbf{0},\label{eq:KKT_GLasso_5}
\end{align}
where $\{\pazocal{G}_{i}^{\star}\}_{i=1}^{S}$ are the local data selectors satisfying all the previous conditions, while (see~\cite{friedman2010note})
\begin{equation}\label{eq:omega_def}
\omega_{i}^{\star}=\begin{cases}
\frac{\pazocal{G}_{i}^{\star}}{\|\pazocal{G}_{i}^{\star}\|_{2}}, &\mbox{ if } \pazocal{G}_{i}^{\star}\neq \mathbf{0},\\
v, &\mbox{ if } \pazocal{G}_{i}^{\star}= \mathbf{0},
\end{cases}
\end{equation}
and $v\in \mathbb{R}^{N_i-L-\rho+1}$ is such that $\|v\|_{2}\leq 1$, while 
 \begin{align}
 & W_{F}^{i}=2(U_{F}^{i})^{\top}\pazocal{R}U_{F}^{i}+2(Y_{F}^{i})^{\top}\pazocal{Q}Y_{F}^{i},\label{eq:KKT_utils1} \\
 & c_{\neq i}=2(U_{F}^{i})^{\top}\pazocal{R}\varepsilon_{u}^{i}+2(Y_{F}^{i})^{\top}\pazocal{Q}\varepsilon_{y}^{i}, \label{eq:KKT_utils2}
 \end{align}
 with $\varepsilon_{u}^{i}=\sum_{\substack{j=1\\j\neq i}}^{S}(U_{F}^{j}\pazocal{G}_{j}^{\star}-\boldsymbol{u^{\mathrm{o}}})$ and  $\varepsilon_{y}^{i}=\sum_{\substack{j=1\\j\neq i}}^{S}(Y_{F}^{j}\pazocal{G}_{j}^{\star}-\boldsymbol{y^{\mathrm{o}}})$.
\end{subequations}

As done in Section~\ref{sec:Elastic_DeePC}, let us now split the Lagrange multipliers $\mu$ associated with the inequality constraints in \eqref{eq:general_deepc2} into two groups $\mu=\begin{bmatrix}
    \tilde{\mu}^{\top} & \hat{\mu}^{\top}
\end{bmatrix}^{\top}$, where $\tilde{\mu}\! >\!\mathbf{0}$ are the Lagrange multipliers linked to the subset of \emph{active} inequality constraints and $\hat{\mu}\!=\!\mathbf{0}$ are those associated with inactive constraints (see \eqref{eq:KKT_GLasso_4}-\eqref{eq:KKT_GLasso_5}). Accordingly, let $\tilde{\Gamma}_{i}$ and $\tilde{\gamma}$ be the active inequalities coefficients, so that \eqref{eq:KKT_GLasso_1} becomes 
\begin{equation}\label{eq:KKT_GLasso_1_bis}
    W_{F}^{i}\pazocal{G}_{i}^{\star}+c_{\neq i}+\lambda_i\omega_{i}^{\star}+(\tilde{Z}_{i})^{\top}\tilde{\alpha}=\!\mathbf{0},~\forall i \!\in\! [S],
\end{equation}
where $\tilde{Z}_{i}$ is defined as in \eqref{eq:equality_utils} and $\tilde{\alpha}=\begin{bmatrix}
    \alpha^{\top} & \tilde{\mu}^{\top}
\end{bmatrix}^{\top}$, while \eqref{eq:KKT_GLasso_2} can be expanded as in \eqref{eq:DeePC_init_constr_ext}.
Under Assumption~\ref{assumption:technical1}, suppose that the $\iota$-th local selector $\pazocal{G}_{\iota}^{\star}$, with $\iota \in [S]$, is shrunk to zero when solving \eqref{eq:general_deepc2} with the Group Lasso regularizer, i.e., $\pazocal{G}_{\iota}^{\star}=\mathbf{0}$, while the other ones are not. Then, based on \eqref{eq:omega_def}, \eqref{eq:KKT_GLasso_1_bis} for $i=\iota$ can be simplified as
\begin{equation}\label{eq:KKT_GLasso_1_ter}
c_{\neq \iota}+\lambda_\iota v+\tilde{Z}_{\iota}^{\top}\tilde{\alpha}=\!\mathbf{0},
\end{equation}
and, consequently,
\begin{equation}\label{eq:v_iota}
    v=-\frac{1}{\lambda_\iota}\left(c_{\neq \iota}+\tilde{Z}_{\iota}^{\top}\tilde{\alpha}\right).
\end{equation}
Since $\|v\|_{2}\leq 1$, the latter condition implies that the following should hold
\begin{equation}
    \|c_{\neq \iota}+\tilde{Z}_{\iota}^{\top}\tilde{\alpha}\|_{2}\leq \lambda_{\iota},
\end{equation}
resulting in a lower bound on $\lambda$ in \eqref{eq:GLasso} for the $\iota$-th local selector to be shrunk to zero:
\begin{equation}\label{eq:bound_lambdaGLasso}
    \lambda \geq \frac{\|c_{\neq \iota}+\tilde{Z}_{\iota}^{\top}\tilde{\alpha}\|_{2}}{\sqrt{N_{\iota}-L-\rho+1}}. 
\end{equation}
To characterize this condition, we are left to find an explicit expression for $\tilde{\alpha}$. We do this by further assuming that no other group is shrunk, namely $\pazocal{G}_{i}^{\star} \neq \mathbf{0}$ for all $i\neq \iota$ with $i \in [S]$. In turn, this implies that the following should hold for all $i 
\neq \iota$:
\begin{equation}
\left(W_{F}^{i}+\frac{\lambda_i}{\|\pazocal{G}_{i}^{\star}\|_{2}}\right)\pazocal{G}_{i}^{\star}+c_{\neq i}+(\tilde{Z}_{i})^{\top}\tilde{\alpha}=\mathbf{0},    
\end{equation}
leading to
\begin{equation}\label{eq:Gi_star_GLasso}
    \pazocal{G}_{i}^{\star}=-W_i^{-1}(c_{\neq i}+(\tilde{Z}_{i})^{\top}\tilde{\alpha}),~~\forall i\neq \iota,~i \!\in\![S],
\end{equation}
where $W_i=W_{F}^{i}+\frac{\lambda_i}{\|\pazocal{G}_{i}^{\star}\|_{2}}$ is invertible by construction. Let us now assume the following.
\begin{assumption}[Equalities well posedness II]\label{assumption:technical2}
 The matrix 
 \vspace{-2.5mm}
 \begin{equation}
     \bar{W}=\sum_{\substack{i=1\\i\neq \iota}}^{S}\tilde{Z}_i W_i^{-1}\tilde{Z}_i^\top,
      \vspace{-2.5mm}
 \end{equation}
is invertible. 
\end{assumption}
Under this assumption, using condition in \eqref{eq:Gi_star_GLasso} 
to replace $\{\pazocal{G}_{i}^{\star}\}_{\substack{i=1\\i\neq \iota}}^{S}$ into \eqref{eq:DeePC_init_constr_ext}, we ultimately obtain
\begin{equation}\label{eq:tilde_alpha2}
    \tilde{\alpha}=-\bar{W}^{-1}\left(\sum_{\substack{i=1\\i\neq \iota}}^{S}\tilde{Z}_iW_i^{-1}c_{\neq i}+\tilde{b}\right),
\end{equation}
which, substituted into \eqref{eq:bound_lambdaGLasso}, yields the following conclusion.
\begin{proposition}[Local shrinkage \& penalty tuning]\label{prop:local_shrink_lambda}
    Let Assumption~\ref{assumption:technical2} hold for all possible combinations of active constraints. Consider \eqref{eq:general_deepc2} with the CAP regularization in \eqref{eq:GLasso}. A necessary and sufficient condition for its solution $\pazocal{G}_{\iota}^{\star}$, with $\iota \in [S]$, to be shrunk to zero is for the regularization penalty $\lambda$ to satisfy \eqref{eq:bound_lambdaGLasso}, with $\tilde{\alpha}$ given in \eqref{eq:tilde_alpha2}. \hfill $\square$
\end{proposition}
\begin{proof}
    The proof follows from the fact that \eqref{eq:KKT_GLasso} are necessary and sufficient conditions for the optimality of $\{\pazocal{G}_{i}^{\star}\}_{i=1}^{S}$, and from the derivation from \eqref{eq:KKT_GLasso_1_bis} to \eqref{eq:tilde_alpha2}. 
\end{proof}
Looking back at condition \eqref{eq:bound_lambdaGLasso}, then \eqref{eq:tilde_alpha2} indicates that the numerator on the lower bound on $\lambda$ depends on all data and the local selector of other modes than the $\iota$-th one. 
\begin{remark}[Shrinking multiple local selectors]
    The steps from \eqref{eq:KKT_GLasso_1_bis} to \eqref{eq:tilde_alpha2} can be straightforwardly extended to the case where more than one local selector is shrunk to zero. In this case, $\lambda$ will have to satisfy 
    \begin{equation}
        \lambda \geq \max_{j \in J} \frac{\|c_{\neq j}+\tilde{Z}_{j}^{\top}\tilde{\alpha}\|_2}{\sqrt{N_j-L-\rho+1}},
    \end{equation}
     where $J \subset [S]$ is the set of indexes of the shrunk local selectors and $\tilde{\alpha}$ solves \eqref{eq:KKT_GLasso_1}-\eqref{eq:KKT_GLasso_5}.
    \hfill $\square$
\end{remark}
Since $\lambda$ depends on the optimal selectors that are not shrunk through $\tilde{\alpha}$, the lower bound in \eqref{eq:bound_lambdaGLasso} cannot serve as a tuning rule. Therefore, we cannot enforce such a condition for Definition~\ref{def:coherence} to be met a priori. While one could eventually find a configuration such that Definition~\ref{def:coherence} is met, it will never be possible to have coherence when the reference and/or the initial condition require the system to switch, given the lack of structural flexibility of our predictor.

\section{The impact of unknown modes in the data matrices construction}\label{sec:misclassification}
While in the previous sections we have relied on Assumption~\ref{assumption:mode_access}, the latter is not realistic in a purely data-driven setting. Indeed, it implies that one has partial knowledge of the controlled system. We now lift this assumption, and indicate a strategy to estimate the active mode for the available batch of data $\pazocal{D}$. Moreover, we analyze the impact that \emph{misclassification errors} have on the cost attained by the non-coherent DeePC schemes with shrinkage considered in Section~\ref{sec:shrinkage_based}.

Based on \eqref{eq:mode_selection2}, the active mode associated with a data point $(u_{\tau}^{d},y_{\tau}^{d})$ for $\tau \in [N]$, with $\tau$ describing the time index for the data acquisition, can be determined by looking at the past $\rho$ input/output values and the associated control action $u_{\tau}^{d}$. We hence discard the first $\rho$ samples in $\pazocal{D}$ and use an unsupervised clustering technique, e.g., K-means~\cite{hartigan1979k}, to separate the vectors $\smallmat{(y_{[\tau-\rho,\tau-1]}^{d})^{\top} & (u_{[\tau-\rho,\tau]}^{d})^{\top}}^{\top}$ into $S$ clusters $\{\pazocal{K}_{i}\}_{i=1}^{S}$, and estimate the mode sequence
\begin{equation}\label{estimated}
    \hat{s}_{\tau}^{d}=i \iff \begin{bmatrix}
        y_{[\tau-\rho,\tau-1]}^{d} \\
        u_{[\tau-\rho,\tau]}^{d}
    \end{bmatrix} \in \pazocal{K}_{i},~~i \in [S],
\end{equation}
for $\tau \in [\rho,N]$. Such an approach can nonetheless lead to $\hat{s}_{\tau}^{d}\neq s^{d}_{\tau}$ for some $\tau$ even in the absence of noise, as the definition of the clusters $\{\pazocal{K}_{i}\}_{i=1}^{S}$ is only contingent to the finite number of observed data samples~\cite[Section 2]{ferraritrecate2003}. In turn, this could lead to a different partition of the dataset $\pazocal{D}$ with respect to the ideal one, i.e., with respect to the sub-sets $\pazocal{D}_{i}$ defined in \eqref{eq:data_local}, for $i \in [S]$ and, hence, differences in the dimension and composition of the data matrices used to construct the data-driven predictor in \eqref{eq:general_deepc2}.

Let us now replace the predictor in \eqref{eq:DeePC_model}-\eqref{eq:DeePC_affine_constr} with its equivalent Mosaic representation, i.e., 
\begin{equation}\label{eq:mosaic_predictor}
    \sum_{i=1}^{S}\!\!\begin{bmatrix}
        Z_{P}^{i}\\
        U_{F}^{i}\\
        Y_{F}^{i}\\
        \pazocal{I}_i
    \end{bmatrix}\!\pazocal{G}_{i}\!=\!\!\underbrace{\begin{bmatrix} Z_{P}^{1} \!&\! Z_{P}^{2} \!&\! \ldots \!&\! Z_{P}^{S}\\
    U_{F}^{1} \!&\! U_{F}^{2} \!&\! \ldots \!&\! U_{F}^{S}\\
    Y_{F}^{1} \!&\! Y_{F}^{2} \!&\! \ldots \!&\! Y_{F}^{S}\\
    \pazocal{I}_1 & \pazocal{I}_2 \!&\! \ldots \!&\! \pazocal{I}_S
    \end{bmatrix}}_{:=M^d}\!g\!=\!\!\begin{bmatrix}
        Z_P\\
        U_F\\
        Y_F\\
        \pazocal{I}
    \end{bmatrix}\!g\!=\!\!\begin{bmatrix}
        z_{\mathrm{ini}}\\
        u_{[0,L-1]}\\
        y_{[0,L-1]}\\
        \boldsymbol{1}_S
    \end{bmatrix}\!.
\end{equation}
Note that $\hat{\pazocal{I}}$ differs from $\pazocal{I}$ as the characteristics of $\hat{\pazocal{I}}_i$ are dictated by the number $\hat{N}_i$ of data points assigned by the unsupervised clustering technique of choice to the $i$-th mode, which might be different from $N_i$, with $i \in [S]$.

In the presence of classification errors, the Mosaic matrix in \eqref{eq:mosaic_predictor} is actually replaced by
\begin{subequations}\label{eq:permutation}
    \begin{equation}
    \hat{M}^d=\begin{bmatrix} \hat{Z}_{P}^{1} & \hat{Z}_{P}^{2} & \ldots &\hat{Z}_{P}^{S}\\
    \hat{U}_{F}^{1} & \hat{U}_{F}^{2} \!&\! \ldots \!&\! \hat{U}_{F}^{S}\\
    \hat{Y}_{F}^{1} & \hat{Y}_{F}^{2} \!&\! \ldots \!&\! \hat{Y}_{F}^{S}\\
    \hat{\pazocal{I}}_1 & \hat{\pazocal{I}}_2 \!&\! \ldots \!&\! \hat{\pazocal{I}}_S
    \end{bmatrix}=\begin{bmatrix}
        \hat{Z}_{P}\\
        \hat{U}_{F}\\
        \hat{Y}_{F}\\
        \hat{\pazocal{I}}
    \end{bmatrix},
\end{equation}
satisfying
\begin{equation}
    \mathrm{vec}(\hat{M}^d)=P(\hat{s}_{[\rho+1,N]}^{d})\mathrm{vec}(M^d),
\end{equation}
\end{subequations}
where $P(\hat{s}_{[\rho+1,N]}^{d})$ permutes the elements in the \textquotedblleft true\textquotedblright \ Mosaic $M^d$ according to the estimated mode sequence. This leads to a different predictor
\begin{equation}\label{eq:mosaic_predictor_wrong}
    \hat{M}^d\hat{g}=\begin{bmatrix}
        \hat{Z}_P\\
        \hat{U}_F\\
        \hat{Y}_F\\
        \hat{\pazocal{I}}
    \end{bmatrix}\hat{g}=\begin{bmatrix}
        z_{\mathrm{ini}}\\
        u_{[0,L-1]}\\
        y_{[0,L-1]}\\
        \boldsymbol{1}_S
    \end{bmatrix}.
\end{equation}

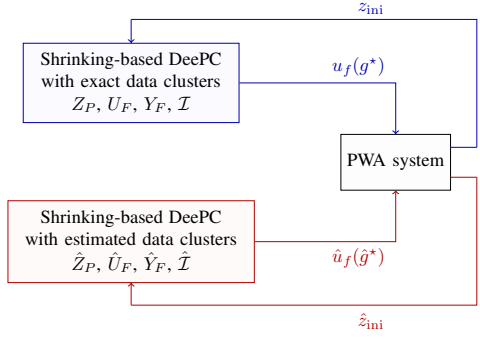
\begin{figure}[!tb]
    \centering
    \scalebox{.7}{
    \begin{tikzpicture}
        \node[draw=blue!70!black,rectangle,fill=blue!2!white] (exact) {\begin{tabular}{c}
             Shrinking-based DeePC\\
             with exact data clusters\\
             $Z_P$, $U_F$, $Y_F$, $\pazocal{I}$
        \end{tabular}};
        \node[draw=red!70!black,rectangle,fill=red!2!white, below of=exact,node distance=3cm] (approx) {\begin{tabular}{c}
             Shrinking-based DeePC\\
             with estimated data clusters\\
             $\hat{Z}_P$, $\hat{U}_F$, $\hat{Y}_F$, $\hat{\pazocal{I}}$
        \end{tabular}};
        \node[coordinate,below of=exact,node distance=1.5cm] (aid1) {};
        \node[coordinate,above of=exact,node distance=1.2cm] (aid2) {};
        \node[coordinate,below of=approx,node distance=1.2cm] (aid3) {};
        \node[draw,rectangle,fill=gray!2!white, right of=aid1,node distance=5cm,minimum width=3em,minimum height=3em] (plant) {PWA system};
        \node[coordinate, right of=plant,node distance=1.25cm] (aid4) {};
        \node[coordinate, above right of=aid4,node distance=.4cm] (aid5) {};
        \node[coordinate, left of=aid5,node distance=.5cm] (aid6) {};
        \node[coordinate, below right of=aid4,node distance=.4cm] (aid7) {};
        \node[coordinate, left of=aid7,node distance=.5cm] (aid8) {};
        \draw[->,red!70!black] (approx) -| node[yshift=-.3cm,xshift=-.7cm]{\textcolor{red!70!black}{$\hat{u}_f(\hat{g}^\star)$}} (plant);
        \draw[->,blue!70!black] (exact) -| node[yshift=.3cm,xshift=-.7cm]{\textcolor{blue!70!black}{$u_f(g^\star)$}} (plant);
        \draw[-,blue!70!black] (aid6) -- (aid5);
        \draw[-,blue!70!black] (aid5) |- node[yshift=.2cm,xshift=-2cm] {\textcolor{blue!70!black}{$z_{\mathrm{ini}}$}} (aid2);
        \draw[->,blue!70!black] (aid2) -| (exact);
        \draw[-,red!70!black] (aid8) -- (aid7);
        \draw[-,red!70!black] (aid7) |- node[yshift=-.3cm,xshift=-2cm] {\textcolor{red!70!black}{$\hat{z}_{\mathrm{ini}}$}} (aid3);
        \draw[->,red!70!black] (aid3) -| (approx);
    \end{tikzpicture}
    }
    \caption{Schematic of the loop between the controlled PWA system and the shrinkage-based data-enabled predictive controllers obtained with exact and estimated data clusters.}
    \label{fig:misclass_loop}
\end{figure}
Given this difference, the solution of \eqref{eq:general_deepc2} with the data matrices constructed through the estimated mode sequence $\hat{g}^{\star}$ differs from the DeePC solution $g^{\star}$ obtained by relying on the true mode sequence one, i.e., 
\begin{equation}\label{eq:epsilon_g}
    \varepsilon_g=\hat{g}^\star-g^\star \neq \mathbf{0}.
\end{equation}
In turn, as schematized in \figurename{~\ref{fig:misclass_loop}}, such a difference causes the optimal control sequences returned by DeePC and its initial conditions to also differ, namely
\begin{subequations}\label{eq:epsilon_other}
    \begin{align}
        & \varepsilon_u=U_Fg^\star-\hat{U}_F\hat{g}^\star=\hat{u}_{f}^{\star}(\hat{g}^\star)-u_f^\star(g^\star)\neq \mathbf{0}, \\
        & \varepsilon_{\mathrm{ini}}=z_{\mathrm{ini}}-\hat{z}_{\mathrm{ini}} \neq \mathbf{0}.
    \end{align}
\end{subequations}
This difference further leads the attained costs, i.e., 
\begin{subequations}\label{eq:costs}
\begin{align}
    \nonumber &J^{\mathrm{reg}}(g^{\star})\!=\!\|y_f(g^{\star})\!-\!\boldsymbol{y^{\mathrm{o}}}\|_{\pazocal{Q}}^{2}\!+\!\|U_Fg^{\star}\!\!-\!\boldsymbol{u^{\mathrm{o}}}\|_{\pazocal{R}}^{2}+r(g^\star)\\
    &\qquad \quad~= J(g^\star)+r(g^\star),\label{eq:J_reg1}\\
    \nonumber &J^{\mathrm{reg}}(\hat{g}^{\star})\!=\!\|\hat{y}_f(\hat{g}^{\star})\!-\!\boldsymbol{y^{\mathrm{o}}}\|_{\pazocal{Q}}^{2}+\|\hat{U}_F\hat{g}^{\star}\!\!-\!\boldsymbol{u^{\mathrm{o}}}\|_{\pazocal{R}}^{2}+\hat{r}(\hat{g}^\star)\\
    &\qquad \quad~= J(\hat{g}^\star)+\hat{r}(\hat{g}^\star),\label{eq:J_reg2}
\end{align} 
\end{subequations}
with $y_f(g^{\star})$ and $\hat{y}_f(g^{\star})$ be the outputs of the PWA system respectively \emph{realized} by feeding it with $u_f(g^\star)=U_Fg^\star$ and $\hat{u}_f(\hat{g}^\star)=\hat{U}_F\hat{g}^\star$, and  
\begin{equation}\label{eq:shrinkage_terms}
    r(g^\star)=\!\!\sum_{j=1}^{M}\lambda_{j}\!\!\sum_{i=1}^{S}r_i^{j}(\pazocal{G}_i^\star),~~~\hat{r}(\hat{g}^\star)=\!\!\sum_{j=1}^{M}\lambda_{j}\!\!\sum_{i=1}^{S}\hat{r}_i^{j}(\hat{\pazocal{G}}_i^\star),
\end{equation}
where $\hat{r}_i^{j}(\hat{\pazocal{G}}_i^\star)$ denotes the local regularizer modified by misclassifications, to differ as formalized in the following Lemma.
\begin{lemma}[Cost of misclassification]\label{lemma:misclass}
    Let $g^\star$ and $\hat{g}^{\star}$ be the finite, optimal solutions of \eqref{eq:general_deepc} using the predictor in \eqref{eq:mosaic_predictor} and \eqref{eq:mosaic_predictor_wrong}, respectively. Suppose that neither of the two resulting optimal input sequences causes the controlled system to become unstable if fed to it for $L$ time instants. Moreover, assume that $\varepsilon_g$ in \eqref{eq:epsilon_g} and $\varepsilon_{\mathrm{ini}}$ and $\varepsilon_{u}$ in \eqref{eq:epsilon_other} satisfy 
    \begin{equation}\label{eq:bounds}
        \|\varepsilon_{g}\|_{1} \leq \eta_g,~~~\|\varepsilon_{\mathrm{ini}}\|_{2}^{2}\leq \eta_{\mathrm{ini}},~~\|\varepsilon_{u}\|_{2}^{2} \leq \eta_u.
    \end{equation}
    Then, when using Elastic-DeePC, the following bound hold on the attained regularized cost:
    \begin{align}\label{eq:bound_elastic}
        \nonumber & J^{\mathrm{reg}}(g^{\star}) \!\leq\! J^{\mathrm{reg}}(\hat{g}^{\star})\!+\!c_1(\eta_{\mathrm{ini}}\!+\!\|\hat{z}_{\mathrm{ini}}\|_{2}^{2})\!+\!c_2\eta_{u}\!+\!c_3^{\mathrm{EN}}\eta_g^2\\
        &\qquad+\!c_4^{\mathrm{EN}}\|\hat{g}^{\star}\|_{2}^{2}\!+\!\lambda_1\eta_g+\!2\bar{\varphi}\hat{\varepsilon}\!+\!8\bar{\varphi}\|\Delta\pazocal{M}\|_{2}^{2},
    \end{align}
    where 
    \begin{align*}
    &c_1=16\bar{\varphi}\|\Delta\pazocal{M}\|_{2}^{2},~~~c_2=\bar{\varphi}(16\|\Delta\pazocal{M}\|_{2}^{2}+2),\\
    & c_3^{\mathrm{EN}\!}\!=16\bar{\varphi}\|Y_F\|_{2}^{2}\!+\!2\lambda_2,\\
    &c_4^{\mathrm{EN}\!}\!=\bar{\varphi}(16\|\Delta\pazocal{M}\|_{2}^{2}\|\hat{U}_{F}\|_{2}^{2}+8\|\Delta Y_F\|_{2}^{2})+\!\lambda_2,
\end{align*}
    where $\Delta \pazocal{M}$ indicates the mismatch between the \textquotedblleft true\textquotedblright \ $L$-steps ahead model of the system and the data-driven subspace predictor associated with \eqref{eq:mosaic_predictor}, $\bar{\varphi}$ denotes the maximum among the eigenvalues of $\pazocal{Q}$ and $\pazocal{R}$ in \eqref{eq:costs}, and 
    \begin{equation}\label{eq:varepsilon_hat}
    \hat{\varepsilon}=\|\hat{Y}_F\hat{g}^\star-\boldsymbol{y^{\mathrm{o}}}\|_{2}^{2}+\|\hat{U}_F\hat{g}^\star-\boldsymbol{u^{\mathrm{o}}}\|_{2}^{2}.
\end{equation}
  When using CAP-DeePC, the bound instead becomes 
 \begin{align}\label{eq:bound_cap}
     \nonumber & J^{\mathrm{reg}}(g^{\star}) \!\leq\! J^{\mathrm{reg}}(\hat{g}^{\star})\!+\!c_1(\eta_{\mathrm{ini}}\!+\!\|\hat{z}_{\mathrm{ini}}\|_{2}^{2})\!+\!c_2\eta_{u}\!+\!c_3^{\mathrm{CAP}}\eta_g^2\\
        &~~+\!c_4^{\mathrm{CAP}}\|\hat{g}^{\star}\|_{2}^{2}+\bar{\lambda}(\eta_g\!+\!\|\hat{g}^\star\|_{1})\!+\!2\bar{\varphi}\hat{\varepsilon}\!+\!8\bar{\varphi}\|\Delta\pazocal{M}\|_{2}^{2},
 \end{align}
 with $\bar{\lambda}=\max_{i \in [S]}\lambda_i$ and
    \begin{equation*}
    c_3^{\mathrm{CAP}\!}\!=16\bar{\varphi}\|Y_F\|_{2}^{2},~c_4^{\mathrm{CAP}\!}\!=\bar{\varphi}(16\|\Delta\pazocal{M}\|_{2}^{2}\|\hat{U}_{F}\|_{2}^{2}+8\|\Delta Y_F\|_{2}^{2}).\!
\end{equation*}
\end{lemma}
\begin{proof}
    The proof can be found in Appendix~\ref{appendix:proof_mismatch}.
\end{proof}
As expected, the bounds in \eqref{eq:bound_elastic}-\eqref{eq:bound_cap} depend on the differences between the ideal initial condition and that attained with misclassified data in the batch $\pazocal{D}$, realized input sequence, and optimal weights. At the same time, both depend on the mismatch $\Delta \pazocal{M}$ between the actual system and the subspace predictor constructed with perfectly clustered data, acknowledging possible modeling errors even in this ideal scenario. Their main difference lies in the dependence of \eqref{eq:bound_cap} on $\lambda_i$ and hence, the amount of data assigned to each mode, making this bound more sensitive to clusters' sizes.


\section{Illustrative example}\label{sec:example}
\begin{table}[!tb]
    \caption{Parameters of shrinking-based DeePC schemes. The regularization of Elastic-DeePC has been chosen for the 1-norm regularization to dominate over the 2-norm one.}
    \label{tab:parameters_DeePC}
    \centering
    \begin{tabular}{cccc|cc|c} 
         & & & \multicolumn{1}{c}{} & \multicolumn{2}{c}{Elastic-DeePC} & CAP-DeePC\\
        \hline
        $L$ & $\rho$ & $Q$ & $R$ & $\lambda_1$ & $\lambda_2$ & $\lambda$\\
        \hline
         19 & 25 & 1 & 1 & 10 & $10^{-9}$ & 10\\
         \hline
    \end{tabular}
\end{table}
\begin{figure}[!tb]
    \centering
    \includegraphics[scale=0.7]{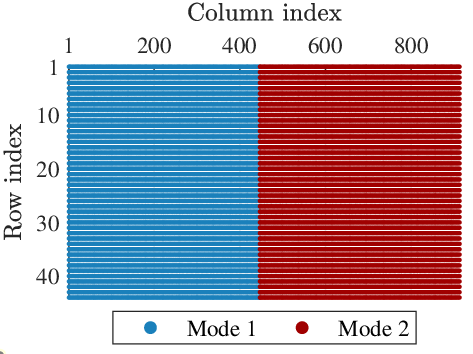}
    \caption{Mosaic data matrix with exact data clustering.}
    \label{fig:hankel_exact}
\end{figure}
Consider the relatively simple PWA system
\begin{equation}\label{eq:example_system}
\begin{aligned}
    x_{t+1}&=\begin{cases}
        -0.3x_t+1.4u_t,~~\mbox{if } x_t<0,~~\texttt{(mode 1)}\\
        0.9x_t+0.15u_t,~~~\mbox{if } x_t\geq 0,~~\texttt{(mode 2)}
    \end{cases}\\
    y_t&=x_t,
\end{aligned}    
\end{equation}
which we aim to control to track two piecewise constant references $\boldsymbol{y}^{\mathrm{o}}$ and the associated piecewise constant reference inputs. These references are chosen to force the system to switch once between its two operating regimes\footnote{The experiments can be reproduced with the code available at \url{https://github.com/GiacomelliGianluca/Shrinkage_Impact_DeePC.git}}. 

To this end, we collect a dataset by exciting the system with an input generated by an oracle model predictive controller (i.e., using the true system model) that weights only the output, with a prediction horizon of $20$. With this approach, we collect $N=1000$ data by tracking a decaying triangular-wave reference with a range of $[-10,10]$, a period of $40$ steps, and an amplitude decreased by $1\%$ each period. While we acknowledge that this data collection is unrealistic, it allows us to obtain a balanced exploration of the two systems' modes during data collection, leading to data subsets that are PE in line with Assumption \ref{assumption:dataset_features}. In turn, this results in the Hankel data structure shown in \figurename{~\ref{fig:hankel_exact}} when data are perfectly clustered. Using this data, we test both shrinking-based DeePC schemes analyzed in Section~\ref{sec:shrinkage_based} with the parameters reported in \tablename{~\ref{tab:parameters_DeePC}}. For both these schemes, we heuristically select the past horizon length $\rho=25$ based on Akaike's Information Criteria (AIC)\footnote{This approach is intended for linear systems subject to and does not account for possible mode switches.} (see \cite{hirotugu1969fitting}). Note that we omit \eqref{eq:DeePC_output_constr} as the controlled system is PieceWise Linear (PWL), and we assume to know this a priori. 
To assess the performance of these strategies, we use the Root Mean Square tracking Errors (RMSE) for the input and the output, i.e.,
\begin{equation}\label{eq:RMSE}
\begin{aligned}
    \mathrm{RMSE}_{\xi}&=\sqrt{\frac{1}{T}\sum_{t=0}^{T-1}(\xi_t\!-\!\xi_t^{\mathrm{o}})^2},
\end{aligned}
\end{equation}
where $\xi$ is a placeholder for $u$ and $y$, and $T=50$ is the length of the considered closed-loop simulation. Moreover, inspired by \cite[Section IV.C]{dorfler2022bridging}, we introduce the Behavioral Performance Indicator (BPI), defined as
\begin{equation}\label{eq:BPI}
    \mathrm{BPI}_t^i=\frac{\|\pazocal{G}_{i,t}\|_{0}}{n_un_t^i +n_x}\geq 0,~~~i \in [S],
\end{equation}
where $\|\pazocal{G}_{i,t}\|_0$ is the number of nonzero entries of the solution of \eqref{eq:general_deepc} at time $t$, while $n_{t}^{i}$ is the number of samples actually belonging to the $i$-th mode in $z_{\mathrm{ini},t}$ and over the future horizon $L$, for all $i \in [S]$. This index thus provides an indication of the distance between the sparsity pattern of the shrinking-based DeePC problems' solutions and that of the actual behavioral representation of the system (see Section~\ref{sec:behaviors}). We point out that \eqref{eq:BPI} represents the BPI for PWL systems, while the same expression holds for PWA ones by adding a one to the denominator due to Lemma \ref{lemma:behaviour_dimension}. If $\mathrm{BPI}_t^i=1$, then the right amount of data associated with the $i$-th local behavior is used to predict the output, while the mode is overrepresented if $\mathrm{BPI}_t^i>1$ and underrepresented otherwise.   
\begin{table*}[!tb]
    \caption{RMSEs attained with Elastic-DeePC and CAP-DeePC with exact and inexact data clustering.}
    \label{tab:performance_indexes}
    \centering
    \begin{tabular}{|c|cc|cc|cc|cc|}
    \cline{2-9}
         \multicolumn{1}{c|}{} &  \multicolumn{4}{c|}{Case 1} & \multicolumn{4}{c|}{Case 2}\\
         \cline{2-9}
          \multicolumn{1}{c|}{} & \multicolumn{2}{c|}{Exact} & \multicolumn{2}{c|}{Misclassified} & \multicolumn{2}{c|}{Exact} & \multicolumn{2}{c|}{Misclassified}\\
         \cline{2-9}
         \multicolumn{1}{c|}{} & Elastic-DeePC & CAP-DeePC & Elastic-DeePC & CAP-DeePC & Elastic-DeePC & CAP-DeePC & Elastic-DeePC & CAP-DeePC\\
         \hline
         RMSE$_u$& \textbf{0.38} & 0.59 & 0.52 & 0.88 & \textbf{0.34} & 0.58 & 0.47 & 0.99 \\
         \hline
         RMSE$_y$ & 2.94 & \textbf{2.46} & 2.88 & 2.78 & 5.37 & \textbf{5.11} & 5.24 & 5.71\\
         \hline
    \end{tabular}
\end{table*}
\begin{figure}[!tb]
    \centering
    \begin{tabular}{c}
         \subfigure[Case 1: transition from mode 1 to mode 2.]{\begin{tabular}{cc}
            \includegraphics[scale=.5]{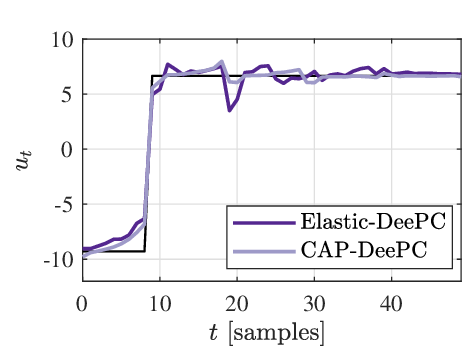}  & \includegraphics[scale=.5]{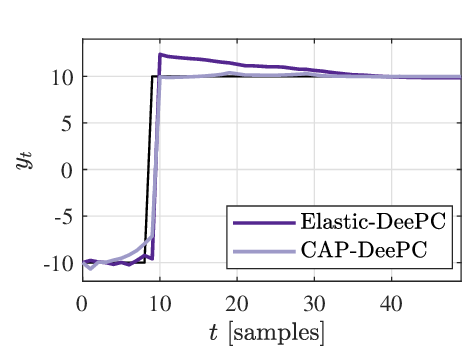} 
         \end{tabular}}\\
         \subfigure[Case 2: transition from mode 2 to mode 1.]{\begin{tabular}{cc}
           \includegraphics[scale=.5]{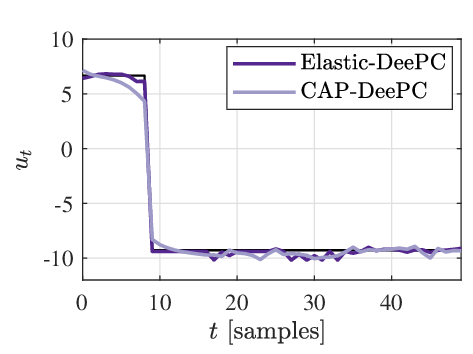} & \includegraphics[scale=.5]{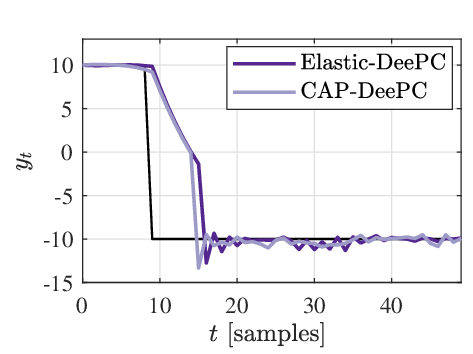} 
         \end{tabular}}
    \end{tabular}
    \caption{Tracking performance achieved with Elastic-DeePC and CAP-DeePC with exact data clustering.}
    \label{fig:exact_clusters_trajectories}
\end{figure}
\begin{figure}[!tb]
    \centering
    \begin{tabular}{c}
         \subfigure[Case 1: transition from mode 1 to mode 2.]{\begin{tabular}{cc}
            \includegraphics[scale=.5]{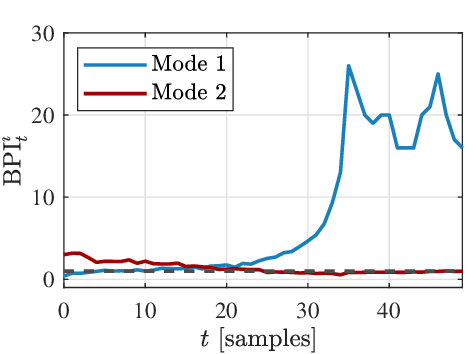}  & \includegraphics[scale=.5]{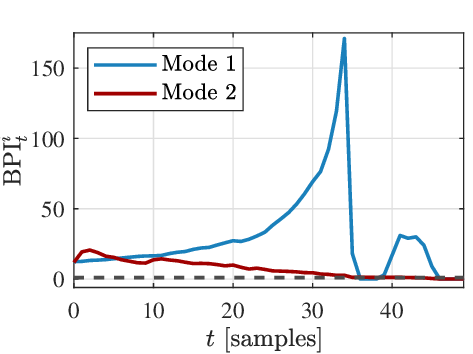}\\
            Elastic-DeePC & CAP-DeePC
         \end{tabular}}\\
         \subfigure[Case 2: transition from mode 2 to mode 1.]{\begin{tabular}{cc}
           \includegraphics[scale=.5]{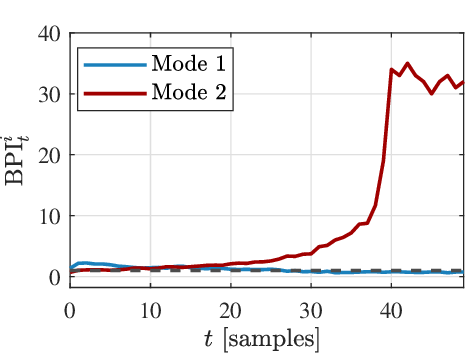} & \includegraphics[scale=.5]{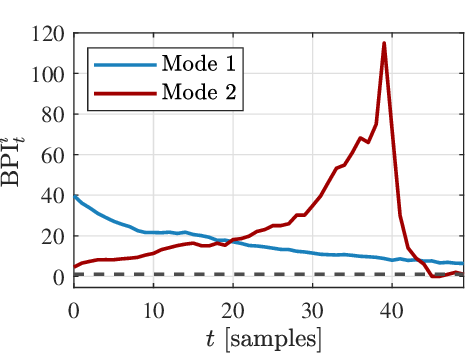}\\
            Elastic-DeePC & CAP-DeePC 
         \end{tabular}}
    \end{tabular}
    \caption{BPIs achieved with Elastic-DeePC and CAP-DeePC with exact data clustering.}
    \label{fig:BPIs}
\end{figure}
\begin{figure}[!tb]
    \centering
    \includegraphics[scale=0.7]{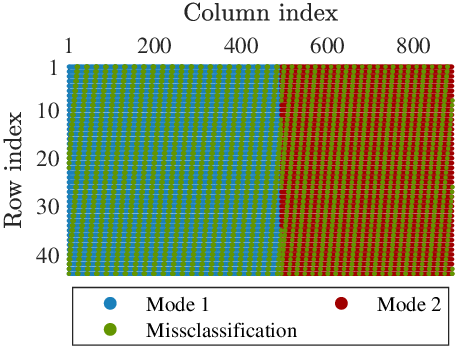}
    \caption{Mosaic data matrix with estimated data clusters.}
    \label{fig:hankel_estimated}
\end{figure}

Let us first assume that the data in $\pazocal{D}$ have been perfectly clustered, leading to the Mosaic matrix depicted in \figurename{~\ref{fig:hankel_exact}}. In this scenario, as shown in \figurename{~\ref{fig:exact_clusters_trajectories}} and confirmed by the RMSEs achieved in \tablename{~\ref{tab:performance_indexes}}, CAP-DeePC tends to perform slightly better than Elastic-DeePC in terms of output tracking while achieving slightly worse input tracking performance, but ultimately leading to a better trade-off between input and output tracking performance. However, irrespective of the considered shrinking scheme, performance tends to deteriorate when the reference forces the system to switch from the second to the first operating mode. This result can be explained by looking at the Behavioral Performance Indicator over time, reported in \figurename{~\ref{fig:BPIs}}. In all cases, the BPIs have a value that hardly corresponds to $1$, ultimately indicating that the solution depends on data that do not match with the actual mode of the system.     

\begin{figure}[!tb]
    \centering
    \begin{tabular}{c}
         \subfigure[Case 1: transition from mode 1 to mode 2.]{\begin{tabular}{cc}
            \includegraphics[scale=.5]{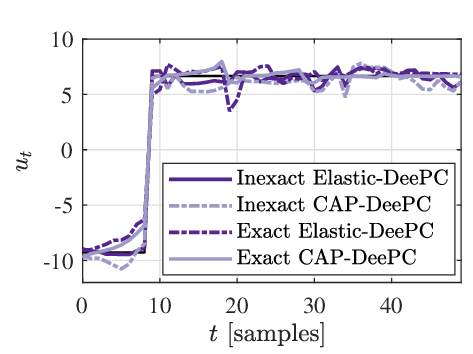}  & \includegraphics[scale=.5]{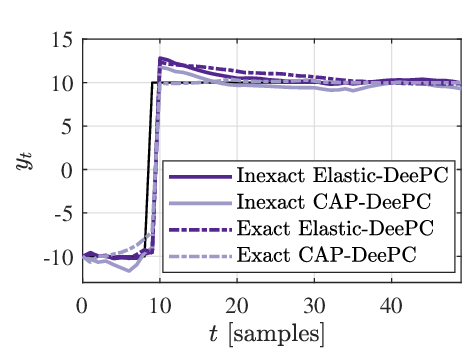} 
         \end{tabular}}\\
         \subfigure[Case 2: transition from mode 2 to mode 1.]{\begin{tabular}{cc}
           \includegraphics[scale=.5]{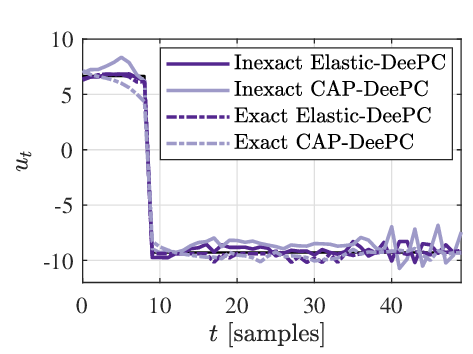} & \includegraphics[scale=.5]{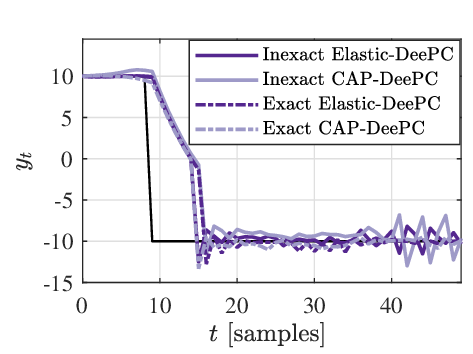} 
         \end{tabular}}
    \end{tabular}
    \caption{Tracking performance achieved with Elastic-DeePC and CAP-DeePC with exact (bold lines) and inexact (dashed lines) data classification.}
    \label{fig:inexact_clusters_trajectories}
\end{figure}

We then focus on the implications of data misclassifications on the performance of Elastic-DeePC and CAP-DeePC, by considering the Mosaic matrix in \figurename{~\ref{fig:hankel_estimated}} built from data classified with K-means. Note that such a procedure leads to a misclassification rate of 14\%. As shown in \figurename{~\ref{fig:inexact_clusters_trajectories}}, both schemes still allow the system to switch between the two modes, yet with trajectories that are different from those achieved with a Mosaic with exactly classified data. In particular, the performance of Elastic-DeePC tends to slightly improve in the misclassified scenario, while CAP-DeePC performance deteriorates when data are misclassified (see \tablename{~\ref{tab:performance_indexes}). This result indicates CAP-DeePC is more sensitive to misclassification errors, in line with Lemma~\ref{lemma:misclass}.  

\section{Conclusions}\label{sec:conclusions}
Building on existing results on PWA input/output realizations, we introduced a behavioral representation of PWA systems and extended Willems' fundamental lemma to define a data-driven multi-step predictor for this class of systems. By considering two DeePC schemes, each featuring a term-wise and a group-wise shrinkage regularizer, we have shown that regularization alone is likely insufficient to achieve coherence between the control policy and the behavior of the controlled system. As the derived data-driven predictor relies on clustered data, we have also analyzed the impact of misclassifications on the losses achieved in the closed-loop for the two considered regularization strategies. Our numerical analysis supports our theoretical insights, showing the inconsistency of shrinking-based DeePC schemes with the PWA behavior as well as the deterioration induced by data misclassifications.  

To ensure the coherence of the DeePC's solution with the local system behavior, future research will focus on exploiting the derived data-driven predictor to guarantee the consistency of DeePC with the PWA behavior by design, at the price of shifting to a mixed-integer DeePC scheme. In addition, future work will explore strategies to lift the assumptions of having noise-free data and access to the actual number of modes of the PWA system.   

\appendix
\subsection{Useful notions on PWA systems \& their properties}\label{appendix:useful_notions}
Here, we recall some definitions that shed light on basic concepts about PWA systems, exploited throughout the paper.
\begin{definition}[Controllability~\cite{bemporad2000}]
	Let $\mathcal{X}_0\! \subseteq\! \mathbb{R}^{n_x}$ and $\mathcal{X}_T \subseteq \mathbb{R}^{n_x}$ be two non-empty sets of initial and \textquotedblleft final\textquotedblright \ states. The PWA system in \eqref{eq:PWA_sys} is said to be controllable in $T$ steps from $\mathcal{X}_0$ to $\mathcal{X}_T$ if for all $x_0 \in \mathcal{X}_0$ there exists an admissible input sequence $u_{[0,T-1]}$ such that $x_T \in \mathcal{X}_T$. 
\end{definition}
\begin{definition}[Finite-time observability~\cite{paoletti2009input}]\label{def:finite_time_obs}
	The system \eqref{eq:PWA_sys} is observable in finite time $T$ if there exists $T \!\in\! \mathbb{Z}_{\geq 0}$ such that the mapping from $x_0$ and $u_{[0,T-1]}$ to $u_{[0,T-1]}$ and $y_{[0,T-1]}$ is invertible.   
\end{definition}
We then characterize a feasible active mode sequence, used in the second technical condition (see Assumption~\ref{ass:condition_C3}).
\begin{definition}[State-space feasible mode sequence~\cite{paoletti2009input}]\label{def:feasible_modes}
	Given $\eta\! \in\! \mathbb{Z}_{>0}$, the mode sequence $\tilde{\sigma}_{[0,\eta-1]}\! \in\! [n_\sigma]^{\eta}$ is feasible for \eqref{eq:PWA_sys} if there exists an initial state $x_0 \!\in\! \mathbb{R}^{n_x}$ and an input sequence $u_{[0,\eta-1]}$ such that $\sigma_{[0,\eta-1]}\!=\!\tilde{\sigma}_{[0,\eta-1]}$.  
\end{definition}
We also define an input/output trajectory for a state-space PWA system, a PWARX, and their equivalence~\cite{paoletti2009input}, enabling our shift from a model-based to a behavioral setting.
\begin{definition}[Trajectory of \eqref{eq:PWA_sys}]\label{def:trajectory_PWA}
	The pairs $\{(u_t,y_t)\}_{t \in \mathbb{Z}_{\geq 0}}$ constitute an input-output trajectory of a PWA system described by \eqref{eq:PWA_sys} if there exists an initial state $x_0 \in \mathbb{R}^{n_x}$ such that \eqref{eq:PWA_sys} is satisfied for all $t \in \mathbb{Z}_{\geq 0}$.
\end{definition}
\begin{definition}[Trajectory of \eqref{eq:PWARX_sys}]\label{def:trajectory_PWARX}
	The pairs $\{(u_t,y_t)\}_{t \in \mathbb{Z}_{\geq 0}}$ constitute an input-output trajectory of a PWA system described by \eqref{eq:PWARX_sys} if \eqref{eq:PWA_sys} is satisfied for all $t \in \mathbb{Z}_{\geq 0}$, with $t \geq \max(n_a,n_b)$.
\end{definition}  

\subsection{Proof of Lemma~\ref{lemma:behaviour_dimension}}\label{appendix:proof_dimension_behavior}
Let $w_{[0,L-1]}=\mathrm{col}(y_{[0,L-1]},u_{[0,L-1]},\boldsymbol{1}_{L}) \in \mathcal{B}^{\mathrm{PWA}\!}_{s|L}$ be the vector comprising the inputs and outputs of \eqref{eq:PWARX_sys} compatible with a given feasible mode sequence $s_{[0,L-1]}$ (see Definition~\ref{def:feasible_modes2}), with $L\geq \rho$. As \eqref{eq:PWARX_sys} is input-output equivalent to \eqref{eq:PWA_sys} according to Definition~\ref{def:equivalence}, then there exists $x_0 \in \mathbb{R}^{n_x}$ and $\sigma_{[0,L-1]} \in [n_\sigma]^{L}$ satisfying \eqref{eq:active_mode} such that the following holds
\begin{equation*}
    y_{[0,L-1]}=\pazocal{O}_0^{L-1}x_0+\pazocal{T}_0^{L-1}u_{[0,L-1]}+\pazocal{C}_0^{L-1},
\end{equation*}
or, equivalently,
\begin{equation}\label{eq:matrix_expression}
    \begin{bmatrix}
        y_{[0,L-1]}\\
        u_{[0,L-1]}\\
        \boldsymbol{1}_{L}
    \end{bmatrix}\!\!=\!\!\underbrace{\begin{bmatrix}
        \pazocal{O}_0^{L-1} \!&\! \pazocal{T}_0^{L-1} \!&\! \pazocal{C}_0^{L-1}\\
    \mathbf{0} \!&\! I \!&\! \mathbf{0}\\
    \mathbf{0} \!&\! \mathbf{0} \!&\! \boldsymbol{1}_{L}
    \end{bmatrix}}_{=\pazocal{B}^{\mathrm{PWA}}}\!
    \begin{bmatrix}
        x_{0}\\
        u_{[0,L-1]}\\
        1
    \end{bmatrix},
\end{equation}
where
\begin{align*}
   &\pazocal{O}_0^{L-1\!}=\begin{bmatrix}
        C_{\sigma_{0}}^{\top} \!&\!
        (C_{\sigma_{1}}A_{\sigma_{0}})^{\!\top\!} \!&\!
        \cdots \!\!&\!
        (C_{\sigma_{L-1}} \pazocal{A}_0^{L-2})^{\!\top}
    \end{bmatrix}^{\!\top}\!\!,\\
    &\pazocal{T}_0^{L-1\!\!} =\!\! \begin{bmatrix}
        D_{\sigma_{0}} & \cdots & \mathbf{0} & \mathbf{0}\\
        C_{\sigma_{1}}B_{\sigma_{0}} & \cdots & \mathbf{0} & \mathbf{0}\\
        \vdots &\ddots & \vdots &\vdots\\
        C_{\sigma_{L-2}}\pazocal{A}_1^{L-3}B_{\sigma_0} & \cdots &
        D_{\sigma_{L-2}} & \mathbf{0}\\     
        C_{\sigma_{L-1}}\pazocal{A}_1^{L-2}B_{\sigma_0} & \cdots &
        C_{\sigma_{L-1}}B_{\sigma_{L-1}} &  D_{\sigma_{L-1}}
    \end{bmatrix},\\
    &\pazocal{C}_0^{L-1\!}\!=\!\begin{bmatrix}
         g_{\sigma_0}^{\top}\!&\!
        (g_{\sigma_1}\!+\!C_{\sigma_{1}}f_{\sigma_{0}})^{\!\top\!} \!&\!\!
        \cdots \!\!&\!
        (g_{\sigma_{L-1}}\!+C_{\sigma_{L-1}} \pazocal{F}_0^{L-2})^{\!\top}
    \end{bmatrix}^{\!\top}\!\!,
\end{align*}
with 
\begin{equation*}
\pazocal{A}_\tau^{\tau'}=\prod_{j=0}^{\tau'-\tau} A_{\sigma_{\tau'-j}},~~~   \pazocal{F}_\tau^{\tau'}=\sum_{h=0}^{\tau'-\tau}\prod_{j=1}^{h} A_{\sigma_{\tau'-j+1}}f_{\sigma_{\tau'-h}}, 
\end{equation*}
for all $\tau,\tau' \in [0,L-1]$ with $\tau\leq\tau'$. As \eqref{eq:PWA_sys} is observable in $L$ steps by assumption, $\mathrm{rank}(\pazocal{O}_0^{L-1})=n_x$ and, as a consequence, $\mathrm{rank}(\pazocal{B}^{\mathrm{PWA}})=n_x+n_u L+1$ (see \eqref{eq:matrix_expression}). Meanwhile, thanks to the input-output equivalence between \eqref{eq:PWA_sys} and \eqref{eq:PWARX_sys}, $\pazocal{B}^{\mathrm{PWA}}$ is a basis for $\mathcal{B}^{\mathrm{PWA}\!}_{s|L}$, leading to \eqref{eq:dimension_behavior} and, thus, concluding the proof.

\subsection{Proof of Theorem~\ref{thm:fundamental_PWA}}\label{appendix:fundamental_lemma}
Let $\mathcal{B}^{+\!}_{|L}$ be the sum of the local, truncated behaviors $\mathcal{B}_{i|L}^{\mathrm{A}}$ in \eqref{eq:local_behavior}, i.e.,
\begin{equation}                 \mathcal{B}^{+\!}_{|L}=\sum_{i=1}^{S}\mathcal{B}_{i|L}^{\mathrm{A}}.
\end{equation}
Under Assumption~\ref{assumption:dataset_features} and based on \cite[Proof of Theorem 1]{fazzi2023addition} it holds that
\begin{equation*}
    \forall~\{\tilde{u}_{[0,L-1]},\tilde{y}_{[0,L-1]}\} \!\in\!\mathcal{B}^{+\!}_{|L},~~\exists g \mbox{ s.t. } \eqref{eq:tilde_data_driven_behavior2}, \eqref{eq:local_affine_equalities} \mbox{ hold,}  
\end{equation*}
where the elements of $\tilde{u}_{[0,L-1]}$ and $\tilde{y}_{[0,L-1]}$ are defined as in \eqref{eq:utils_datarep}, and $g$ is given by \eqref{eq:g_decomposition}. Since $\mathcal{B}^{+\!}_{|L}$ does not account for any switching sequence, this behavior describes trajectories that may not be consistent with a feasible switching sequence $s \in \mathcal{B}_{[S]|L}^{\mathrm{PWA}}$. Therefore, $\mathcal{B}^{+\!}_{|L}$ overapproximates $\mathcal{B}_{s|L}^{\mathrm{PWA}}$, i.e., $\mathcal{B}^{\mathrm{PWA}\!}_{s|L} \subseteq\mathcal{B}^{+\!}_{|L}$. Such an approximation is removed in \eqref{eq:data_driven_behavior}, which can be further recast as
\begin{equation}\label{eq:DD_proof1}
    \begin{bmatrix}
        \tilde{\pazocal{S}}^{1} \!&\! \cdots \!&\! \tilde{\pazocal{S}}^{S}    \end{bmatrix}\odot\begin{bmatrix}\pazocal{M}_L(u_{[1,N]}^{d})\\
			\pazocal{M}_L(y_{[1,N]}^{d})
		\end{bmatrix}g=\begin{bmatrix}
		    u_{[0,L-1]}\\
                y_{[0,L-1]}
		\end{bmatrix},
\end{equation}
with $\tilde{\pazocal{S}}^{i}$ defined in \eqref{eq:S_tilde} and the Mosaic matrices $\pazocal{M}_L(u_{[1,N]}^{d})$ and $\pazocal{M}_L(y_{[1,N]}^{d})$ given in \eqref{eq:Mosaic_modes_dec}. By augmenting \eqref{eq:DD_proof1} to account for \eqref{eq:local_affine_equalities}, we can then prove that the following equality on the image holds:
\begin{equation}
    \mathrm{img}\left(\begin{bmatrix}
        \tilde{\pazocal{S}}^{1} \!\!\!&\!\!\! \cdots \!\!\!&\!\!\!  \tilde{\pazocal{S}}^{S}\\
        \!\!\!&\!\!\!\boldsymbol{1}_{N-SL+S}^{\top}\!\!\!&\!\!\!
    \end{bmatrix}\!\!\odot\!\!\begin{bmatrix}
        \pazocal{M}_L(u_{[1,N]}^{d})\\
			\pazocal{M}_L(y_{[1,N]}^{d})\\
        \boldsymbol{1}_{N-SL+S}^{\top}
    \end{bmatrix}\right)\!=\!\mathcal{B}_{s|L}^{\mathrm{PWA}},
\end{equation}
since, under Assumption~\ref{assumption:dataset_features}, the dimension of $\mathcal{B}_{s|L}^{\mathrm{PWA}}$, i.e., $n_x+n_uL+1$, is matched by the data-driven representation. Hence, for input/output trajectories satisfying \eqref{eq:DD_proof1} it holds $\{u_{[0,L-1]}, y_{[0,L-1]}\}\in \mathcal{B}^{\mathrm{PWA}\!}_{s|L}$. At the same time, based on definition \eqref{eq:fixed_s_behavior_truncated}, this trajectory is associated with a feasible mode sequence $s_{[0,L-1]}$, which is characterized by the data-driven representation using  \eqref{eq:G_dependent_assignment}, concluding the proof. 
    
\subsection{Proof of Proposition~\ref{prop:explicit_elastic}}\label{appendix:proof_prop_elastic}
Following the steps of~\cite[Section III]{giacomelli2025insights}, let us decompose the selectors associated with each mode $\{\pazocal{G}_i\}_{i=1}^{S}$ as follows
\begin{equation}
    \pazocal{G}_i=\pazocal{G}_i^{+}-\pazocal{G}_i^{-},
\end{equation}
with $\pazocal{G}_i^{+},\pazocal{G}_i^{-} \geq \mathbf{0}$. Accordingly, \eqref{eq:general_deepc2} with \eqref{eq:elastic_reg} can be rewritten as a quadratic program in $\{\pazocal{G}_i^{+},\pazocal{G}_i^{-}\}_{i=1}^{S}$ by adding the associated non-negativity constraints and equivalently rewriting the lasso term in \eqref{eq:elastic_reg} as
\begin{equation}\label{eq:lasso_reg_bis}
\lambda_{1}\sum_{i=1}^{S}\|\pazocal{G}_{i}\|_{1}=\lambda_1\sum_{i=1}^{S}\boldsymbol{1}_{N_i-L-\rho+1}^{\top}(\pazocal{G}_i^{+}+\pazocal{G}_i^{-}).
\end{equation}
The Lagrangian \eqref{eq:Lagrangian_elastic} thus becomes
\begin{align}\label{eq:Lagrangian_elastic2}
	\nonumber &\pazocal{L}(\{\pazocal{G}_{i}^{+},\pazocal{G}_{i}^{-}\}_{i=1}^{S},\alpha,\mu;\lambda_{1},\lambda_{2})=J(\{\pazocal{G}_{i}^{+},\pazocal{G}_{i}^{-}\}_{i=1}^{S})\\
	\nonumber & +\!\lambda_{1}\!\sum_{i=1}^{S}\boldsymbol{1}_{N_i-L-\rho+1}^{\top}(\pazocal{G}_{i}^{+\!}\!\!+\!\pazocal{G}_{i}^{-})\!+\!\lambda_{2}\!\sum_{i=1}^{S}\!\|\pazocal{G}_{i}^{+\!}\!-\!\pazocal{G}_{i}^{-}\|_{2}^{2\!}\\
    \nonumber  &+\alpha^{\!\top\!\!}\left(\sum_{i=1}^{S}\tilde{Z}_{P}^{i}(\pazocal{G}_{i}^{+\!}\!-\!\pazocal{G}_{i}^{-})\!-\!\tilde{z}_{\mathrm{ini\!}}\right)\!\!+\!\mu^{\!\top\!\!}\left(\sum_{i=1}^{S}\Gamma_i(\pazocal{G}_{i}^{+\!}\!-\!\pazocal{G}_{i}^{-})\!-\!\gamma\right)\!\\
    & -\sum_{i=1}^{S}(\mu_{i}^{+})^{\top}\pazocal{G}_{i}^{+}-\sum_{i=1}^{S}(\mu_{i}^{-})^{\top}\pazocal{G}_{i}^{-}.
\end{align}
Simplifying the notation by dropping the dimension in the subscript of $\boldsymbol{1}_{N_i-L-\rho+1}$ and the second superscript for $\pazocal{G}_{i}^{+,\star}$ and $\pazocal{G}_{i}^{-,\star}$, the associated KKT conditions are
\begin{subequations}\label{eq:KKT_elastic}
    \begin{align}
        & \bar{W}_{F}^{i}(\pazocal{G}_{i}^{+\!}\!-\!\pazocal{G}_{i}^{-})\!+\!c_{\neq i}\!+\!\lambda_{1}\boldsymbol{1}\!\!+\!(\tilde{Z}_{P}^{i})^{\!\top\!}\alpha\!+\!\Gamma_i^{\top}\mu\!-\!\mu_i^{+\!}\!=\!\mathbf{0},\label{eq:KKT_elastic1}\\
        & -\!\bar{W}_F^{i}(\pazocal{G}_{i}^{+\!}\!-\!\pazocal{G}_{i}^{-})\!-\!c_{\neq i}\!+\!\lambda_{1}\boldsymbol{1}\!\!-\!(\tilde{Z}_{P}^{i})^{\!\top\!}\alpha\!-\!\Gamma_i^{\top\!}\mu\!-\!\mu_i^{\!-\!}\!=\!\mathbf{0},\label{eq:KKT_elastic2}
        \end{align}
    where $\bar{W}_{F}^{i}=W_{F}^{i}+2\lambda_{2}I$, with $W_{F}^{i}$ and $c_{\neq i}$ are defined as in \eqref{eq:KKT_utils1} and \eqref{eq:KKT_utils2}, respectively, for all $i \in [S]$, and 
    \begin{align}
        & \sum_{i=1}^{S} \tilde{Z}_{P}^{i}(\pazocal{G}_{i}^{+}-\pazocal{G}_{i}^{-})-\tilde{z}_{\mathrm{ini}}=\mathbf{0},\label{eq:KKT_elastic3}\\
        & \sum_{i=1}^{S}\Gamma_i(\pazocal{G}_{i}^{+}-\pazocal{G}_{i}^{-})-\gamma \leq \mathbf{0},\label{eq:KKT_elastic4}\\
        & \pazocal{G}_{i}^{+}\geq \mathbf{0},~~\pazocal{G}_{i}^{-}\geq \mathbf{0}, ~~\forall i \in [S],\label{eq:KKT_elastic5}\\
        & \mu^{\top}\left(\sum_{i=1}^{S}\Gamma_i(\pazocal{G}_{i}^{+}-\pazocal{G}_{i}^{-})-\gamma\right)=0,\label{eq:KKT_elastic6}\\
        & (\mu_{i}^{+})^{\top}\pazocal{G}_{i}^{+}=0,~~~\forall i \in [S],\label{eq:KKT_elastic7}\\
        & (\mu_{i}^{-})^{\top}\pazocal{G}_{i}^{-}=0,~~~\forall i \in [S],\label{eq:KKT_elastic8}\\
        & \mu\geq \mathbf{0},~~\mu_{i}^{+}\geq \mathbf{0},~~\mu_{i}^{-}\geq \mathbf{0},~~\forall i \in [S].\label{eq:KKT_elastic9}
    \end{align}
\end{subequations}
We can now split the inequality constraints in \eqref{eq:KKT_elastic4} into the set of \emph{active} and \emph{inactive} inequalities, i.e.,
\begin{subequations}
    \begin{align}
        & \sum_{i=1}^{S}\tilde{\Gamma}_{i}(\pazocal{G}_{i}^{+}-\pazocal{G}_{i}^{-})-\tilde{\gamma}=\mathbf{0}, \label{eq:elastic_equalities}\\
        &\sum_{i=1}^{S}\hat{\Gamma}_{i}(\pazocal{G}_{i}^{+}-\pazocal{G}_{i}^{-})-\hat{\gamma}\leq \mathbf{0},\label{eq:elastic_inequalities}
    \end{align} 
\end{subequations}
and, accordingly, let us split the set of Lagrange multipliers $\mu$ as $\mu=\begin{bmatrix}
    \tilde{\mu}^{\top} & \hat{\mu}^{\top}
\end{bmatrix}^{\top}$, with the first group of multipliers being associated with the active inequalities and the second with the inactive ones. Note that, according to the complementary slackness condition in \eqref{eq:KKT_elastic6}, $\hat{\mu}=\mathbf{0}$, while $\tilde{\mu}>\mathbf{0}$. Moreover, \eqref{eq:KKT_elastic3} can be enlarged as in \eqref{eq:DeePC_init_constr_ext} and $\alpha$ and $\tilde{\mu}$ can be merged into a unique vector of Lagrange multipliers $\tilde{\alpha}$ (see below~\eqref{eq:KKT_GLasso_1_bis}), allowing us to simplify \eqref{eq:KKT_elastic1}-\eqref{eq:KKT_elastic2} as 
\begin{subequations}\label{eq:KKT_lasso12_bis}
       \begin{align}
        & \bar{W}_{F}^{i}(\pazocal{G}_{i}^{+\!}\!-\!\pazocal{G}_{i}^{-})\!+\!c_{\neq i}\!+\!\lambda_{1}\boldsymbol{1}\!\!+\!(\tilde{Z}_{i})^{\!\top\!}\tilde{\alpha}-\!\mu_i^{+\!}\!=\!\mathbf{0},\label{eq:KKT_elastic1_bis}\\
        & -\!\bar{W}_{F}^{i}(\pazocal{G}_{i}^{+\!}\!-\!\pazocal{G}_{i}^{-})\!-\!c_{\neq i}\!+\!\lambda_{1}\boldsymbol{1}\!\!-\!(\tilde{Z}_{i})^{\!\top\!}\tilde{\alpha}-\!\mu_i^{\!-\!}\!=\!\mathbf{0},\label{eq:KKT_elastic2_bis}
        \end{align}
    where $\tilde{Z}_{i}$ is defined as in \eqref{eq:equality_utils}, for all $i \in [S]$. 
\end{subequations}
These two equations can then be summed and subtracted from each other, leading to the following (equivalent) conditions
\begin{subequations}\label{eq:KKT_lasso12_ter}
       \begin{align}
        & \mu_i^{+\!}+\mu_i^{-\!}=2\lambda_{1}\boldsymbol{1}, \label{eq:KKT_elastic1_ter}\\
        & 2\bar{W}_{F}^{i}(\pazocal{G}_{i}^{+\!}\!-\!\pazocal{G}_{i}^{-})\!+2c_{\neq i}\!+\!2(\tilde{Z}_{i})^{\!\top\!}\tilde{\alpha}-\!\mu_i^{\!+\!}+\!\mu_i^{\!-\!}\!=\!\mathbf{0},\label{eq:KKT_elastic2_ter}
        \end{align}
     for all $i \in [S]$. 
\end{subequations}

From \eqref{eq:KKT_elastic2_ter}, we can explicitly find the expression for $\pazocal{G}_{i}^{+\!}\!-\!\pazocal{G}_{i}^{-}$, i.e.,
\begin{equation}\label{eq:Gi_lasso_1}
    \pazocal{G}_{i}^{+\!}-\pazocal{G}_{i}^{-\!}=\!-(\bar{W}_{F}^{i})^{-1}\left[c_{\neq i\!}+\!(\tilde{Z}_{i})^{\top}\tilde{\alpha}-\frac{1}{2}(\mu_{i}^{+\!}-\!\mu_{i}^{-})\right],
\end{equation}
where $\bar{W}_{F}^{i}$ is invertible since $\lambda_2>0$ by assumption. We can then replace this solution into \eqref{eq:DeePC_init_constr_ext} to find the explicit expression for the Lagrange multipliers $\tilde{\alpha}$, i.e.,
\begin{equation}\label{eq:tilde_alpha}
    \tilde{\alpha}\!=\!-\bar{W}_{i}^{-1}\!\!\left[\!\tilde{Z}_i(\bar{W}_F^i)^{-1}c_{\neq i} \!-\! \frac{1}{2}\tilde{Z}_i(\bar{W}_F^i)^{-1}\Delta \mu_i\!-\!d_{\neq i}\!+\!\tilde{b} \right]\!\!,\!\!
\end{equation}
where $\Delta \mu_i=\mu_i^+-\mu_i^-$, $\bar{W}_{i}=\tilde{Z}_i(\bar{W}_{F}^{i})^{-1}\tilde{Z}_{i}^{\top}$ is invertible thanks to Assumption~\ref{assumption:technical1} and
\begin{equation}
    d_{\neq i}=\sum_{\substack{j=1\\j\neq i}}^S \tilde{Z}_j(\pazocal{G}_j^+ - \pazocal{G}_j^-).
\end{equation}
By replacing the expression for $\tilde{\alpha}$ into \eqref{eq:Gi_lasso_1}, similarly to what is done in \cite{he2011lasso}, we can now distinguish three cases for each local selector, i.e., for all $i \in [S]$.
\begin{enumerate}
    \item If $\pazocal{G}_{i}^{+,\star}>\mathbf{0}$, then $\mu_i^{-}>\mathbf{0}$, $\pazocal{G}_{i}^{-,\star}=\mathbf{0}$ and $\mu_i^{+}=\mathbf{0}$. Therefore, according to \eqref{eq:KKT_elastic1_ter}, $\mu_{i}^{-}=2\lambda_{1}\boldsymbol{1}$ and
        \begin{subequations}\label{eq:G_plus_final_lasso}
    \begin{align}
    \nonumber \pazocal{G}_{i}^{+,\star\!}\!=\!-(\bar{W}_{F}^{i})^{-1}&\left[\delta_{i}(c_{\neq i}+\lambda_{1}\boldsymbol{1}_{N_i-L+1})+\right.\\
    &\qquad +\left.\tilde{Z}_{i}^{\top}\bar{W}_{i}^{-1}(d_{\neq i}-\tilde{b})\right],
    \end{align}
    with 
    \begin{align}
    & \delta_{i}=I-\tilde{Z}_{i}^{\top}\bar{W}_{i}^{-1}\tilde{Z}_{i}(\bar{W}_{F}^{i})^{-1}.
    \label{eq:delta_i}
\end{align}
\end{subequations}
    \item If $\pazocal{G}_{i}^{-,\star}>\mathbf{0}$, then $\mu_i^{+}>\mathbf{0}$, $\pazocal{G}_{i}^{+,\star}=\mathbf{0}$ and $\mu_i^{-}=\mathbf{0}$. Therefore, according to \eqref{eq:KKT_elastic1_ter}, $\mu_{i}^{+}=2\lambda_{1}\boldsymbol{1}$ and
    \begin{align}\label{eq:G_minus_final_lasso}
    \nonumber \pazocal{G}_{i}^{-,\star\!}\!=
    (\bar{W}_{F}^{i})^{-1}&\left[\delta_{i}(c_{\neq i}-\lambda_{1}\boldsymbol{1}_{N_i-L+1})+\right.\\
    &\left.\qquad \tilde{Z}_{i}^{\top}\bar{W}_{i}^{-1}(d_{\neq i}-\tilde{b})\right],
    \end{align}
    with 
    $\delta_i$ defined as in \eqref{eq:delta_i}.
    \item If a local selector is shrunk to zero, i.e., $\pazocal{G}_{i}^{+,\star}=\pazocal{G}_{i}^{-,\star}=\mathbf{0}$, then the Lagrange multipliers satisfy
    \begin{equation}\label{eq:multipliers_lasso}
        \frac{1}{2}\delta_{i}(2c_{\neq i}-
        \mu_{i}^{+}+\mu_{i}^{-})=-\tilde{Z}_{i}^{\top}\bar{W}_{i}^{-1}(d_{\neq i}-\tilde{b}).
    \end{equation}
\end{enumerate}
The same reasoning can be extended to all possible combinations of active constraints, leading to $\{\pazocal{G}_{i}^{\star}\}_{i=1}^{S}$ being a PWA function of the initial condition $z_{\mathrm{ini}}$ (which appears in $\tilde{b}$ as shown in \eqref{eq:DeePC_init_constr_ext}), whose value changes over a polyhedral partition of the space of $z_{\mathrm{ini}}$ dictated by \eqref{eq:elastic_inequalities} as well as the auxiliary variables $\{\pazocal{G}_{i}^{+,\star},\pazocal{G}_{i}^{-,\star}\}_{i=1}^{S}$ be non-zero, thus concluding the proof. 
\begin{remark}[Data-dependence]
    Generally, the value of $\pazocal{G}_{i}^{\star}$, with $i\in[S]$ is driven by data from the other modes with non-zero selectors (see \eqref{eq:G_plus_final_lasso}-\eqref{eq:G_minus_final_lasso}). Nonetheless, looking at \eqref{eq:G_plus_final_lasso}-\eqref{eq:multipliers_lasso}, there might exist a specific configuration of $\lambda_1$ and $\lambda_2$ such that only a selector is not null. In this specific case, Definition~\ref{def:coherence} is satisfied because this selector depends only on the data associated with the associated mode. However, coherence cannot still be consistently achieved with this configuration when the reference and/or the initial condition require the system to switch over the prediction horizon.
\end{remark}
\subsection{Proof of Lemma~\ref{lemma:misclass}}\label{appendix:proof_mismatch}
Since $J(\hat{g}^{\star})$ is non-negative by definition, it is straightforward to prove that
\begin{align*}
    J(g^{\star})&\leq J(\hat{g}^{\star})+\|y_f(g^\star)-\boldsymbol{y^{\mathrm{o}}}\|_{\pazocal{Q}}^{2}+\|u_f(g^\star)-\boldsymbol{u^{\mathrm{o}}}\|_{\pazocal{R}}^{2}\\
    &\leq \!J(\hat{g}^{\star})\!+\!\bar{\varphi}(\pazocal{Q})\|y_f(g^\star)\!-\!\boldsymbol{y^{\mathrm{o}}}\|_{2}^{2}\!+\!\bar{\varphi}(\pazocal{R})\|u_f(g^\star)\!-\!\boldsymbol{u^{\mathrm{o}}}\|_{2}^{2}\\
    &\leq \!J(\hat{g}^{\star})\!+\!\bar{\varphi}(\|y_f(g^\star)\!-\!\boldsymbol{y^{\mathrm{o}}}\|_{2}^{2}\!+\!\|u_f(g^\star)\!-\!\boldsymbol{u^{\mathrm{o}}}\|_{2}^{2}),
\end{align*}
where the second and third terms in the first inequality come from the definition of $J(g^\star)$ itself (see \eqref{eq:J_reg1}), while the second inequality is a direct consequence of the properties of quadratic forms, with $\bar{\varphi}(\pazocal{Q})$ and $\bar{\varphi}(\pazocal{R})$ being the maximum eigenvalues of $\pazocal{Q}$ and $\pazocal{R}$, respectively, and $\bar{\varphi}=\max\{\bar{\varphi}(\pazocal{Q}),\bar{\varphi}(\pazocal{R})\}$. 

Let us first consider the last term on the right-hand side of the previous inequality and, specifically, $\|u_f(g^\star)\!-\!\boldsymbol{u^{\mathrm{o}}}\|_{2}^{2}$. This quantity can be bounded as follows
\begin{align}
   \nonumber \|u_f(g^\star)\!-\!\boldsymbol{u^{\mathrm{o}}}\|_{2}^{2}&=\|u_f(g^\star)\!-\!\hat{u}_f(\hat{g}^\star)\!+\!\hat{U}_{F}\hat{g}^\star\!-\!\boldsymbol{u^{\mathrm{o}}}\|_{2}^{2}\\
    \nonumber &=\|\varepsilon_u\!+\!\hat{U}_{F}\hat{g}^\star\!-\!\boldsymbol{u^{\mathrm{o}}}\|_{2}^{2}\\
    \nonumber & \leq 2\|\varepsilon_u\|_{2}^{2}+2\|\hat{U}_{F}\hat{g}^\star\!-\!\boldsymbol{u^{\mathrm{o}}}\|_{2}^{2}\\
    &\leq 2\eta_u+2\|\hat{U}_{F}\hat{g}^\star\!-\!\boldsymbol{u^{\mathrm{o}}}\|_{2}^{2}, \label{eq:bound_u}
\end{align}
where, we have added and subtracted $\hat{U}_{F}\hat{g}^\star$ and used the definition $\hat{u}_f(\hat{g}^\star)=\hat{U}_{F}\hat{g}^\star$ to obtain the first equality, while we have used the definition of $\varepsilon_{u}$ in \eqref{eq:epsilon_other}, its bound in \eqref{eq:bounds}, and the Cauchy-Schwarz inequality to obtain the subsequent inequalities. Meanwhile, $\|y_f(g^\star)\!-\!\boldsymbol{y^{\mathrm{o}}}\|_{2}^{2}$ can be bounded as follows
\begin{align*}
    \|y_f(g^\star)\!-\!\boldsymbol{y^{\mathrm{o}}}\|_{2}^{2}& = \|y_f(g^\star)\!-\!\hat{Y}_F\hat{g}^\star\!+\!\hat{Y}_F\hat{g}^\star\!-\!\boldsymbol{y^{\mathrm{o}}}\|_{2}^{2}\\
    & \leq 2\|y_f(g^\star)\!-\!\hat{Y}_F\hat{g}^\star\|_{2}^{2}\!+\!2\|\hat{Y}_F\hat{g}^\star\!-\!\boldsymbol{y^{\mathrm{o}}}\|_{2}^{2}.
\end{align*}
By focusing on the first term, we can further decompose and upper-bound it as follows
\begin{align*}
    \nonumber &\|y_f(g^\star)\!-\!\hat{Y}_F\hat{g}^\star\|_{2}^{2} = \|y_f(g^\star)\!-\!Y_Fg^\star\!+\!Y_Fg^\star\!-\!\hat{Y}_F\hat{g}^\star\|_{2}^{2}\\
    \nonumber &~~ \leq 2\|y_f(g^\star)\!-\!Y_F\hat{g}^\star\|_{2}^{2}\!+\!2\|Y_Fg^\star\!-\!\hat{Y}_F\hat{g}^\star\|_{2}^{2}\\
    \nonumber &~~ = 2\|y_f(g^\star)\!-\!Y_Fg^\star\!+\!Y_Fg^\star\!-\!Y_F\hat{g}^\star\|_{2}^{2}\!+\!2\|Y_Fg^\star\!-\!\hat{Y}_F\hat{g}^\star\|_{2}^{2}\\
    \nonumber & ~~ \leq 4(\|y_f(g^\star)\!-\!Y_Fg^\star\|_{2}^{2}\!+\!\|Y_F\|_{2}^{2}\|\varepsilon_g\|_{2}^{2})\!+\!2\|Y_Fg^\star\!-\!\hat{Y}_F\hat{g}^\star\|_{2}^{2}\\
     & ~~ \leq 4(\|y_f(g^\star)\!-\!Y_Fg^\star\|_{2}^{2}\!+\!\|Y_F\|_{2}^{2}\eta_g^{2})\!+\!2\|Y_Fg^\star\!-\!\hat{Y}_F\hat{g}^\star\|_{2}^{2},
\end{align*}
where all the inequalities are obtained via the Cauchy-Schwarz inequality, yet the second one relies on the definition of $\varepsilon_g$ in \eqref{eq:epsilon_g}, while the third inequality leverages the fact that $\|\varepsilon_g\|_{2}\leq \|\varepsilon_g\|_{1}\leq \eta_{g}$ according to \eqref{eq:bounds}. The last term on the right-hand side of the previous inequality can be further upper-bounded as follows
\begin{align*}
    \|Y_Fg^\star-\hat{Y}_F\hat{g}^\star\|_{2}^{2}&=\|Y_Fg^\star\!-\!Y_F\hat{g}^\star\!+\!Y_F\hat{g}^\star\!-\!\hat{Y}_F\hat{g}^\star\|_{2}^{2}\\
    & \leq 2\|Y_F\|_{2}^{2}\|\varepsilon_g\|_{2}^{2}+2\|\Delta Y_F\|_{2}^{2}\|\hat{g}^\star\|_{2}^{2}\\
    & \leq 2\|Y_F\|_{2}^{2}\eta_g^{2}+2\|\Delta Y_F\|_{2}^{2}\|\hat{g}^\star\|_{2}^{2},
\end{align*}
where $\Delta Y_F=Y_F-\hat{Y}_F$. Meanwhile, by relying on \eqref{eq:matrix_expression}, the following holds
\begin{equation*}
    y_f(g^\star)=\begin{bmatrix}
        \pazocal{V}_{0}^{L-1} & \pazocal{T}_{0}^{L-1} & \pazocal{C}_{0}^{L-1} 
    \end{bmatrix}\begin{bmatrix}
        z_{\mathrm{ini}}\\
        u_f(g^{\star})\\
        1
    \end{bmatrix},
\end{equation*}
where $\pazocal{V}_{0}^{L-1}$ maps the impact of past input/output sequences of length $\rho$ of the controlled system on its future, whose existence is guaranteed by the system's finite-time observability (see Definition~\ref{def:finite_time_obs} and \cite{paoletti2009input}), due to  $\rho=\mathrm{max}\{n_a,n_b\}$, and that \eqref{eq:PWA_sys} and \eqref{eq:PWARX_sys} are input-output equivalent. At the same time, using the data-driven prediction model in \eqref{eq:mosaic_predictor}, we can equivalently cast $Y_Fg^\star$ as
\begin{equation*}
    Y_Fg^\star=\underbrace{Y_F\begin{bmatrix}
        Z_P\\
        U_F\\
        \pazocal{I}
    \end{bmatrix}^{\dagger}}_{=\Phi}\begin{bmatrix}
        z_{\mathrm{ini}}\\
        u_f(g^\star)\\
        1
    \end{bmatrix},
\end{equation*}
leading to
\begin{align*}
  &\|y_f(g^\star)\!-\!Y_Fg^\star\|_{2}^{2}\leq \|\Delta \pazocal{M}\|_{2}^{2}(\|z_{\mathrm{ini}}\|_{2}^{2}+\|u_{f}(g^\star)\|_{2}^{2}+1)\\
  &~~\leq\! \|\Delta \pazocal{M}\|_{2}^{2\!}(2\|\varepsilon_{\mathrm{ini}}\|_{2}^{2}\!+\!2\|\hat{z}_{\mathrm{ini}}\|_{2}^{2\!}\!+\!2\|\varepsilon_u\|_{2}^{2\!}\!+\!2\|\hat{U}_F\|_{2}^{2\!}\|\hat{g}^{\star}\|_{2}^{2\!}\!+\!1\!)\!\!\\
  &~~\leq\! \|\Delta \pazocal{M}\|_{2}^{2}(2\eta_{\mathrm{ini}}\!+\!2\|\hat{z}_{\mathrm{ini}}\|_{2}^{2}\!+\!2\eta_u\!+\!2\|\hat{U}_F\|_{2}^{2}\|\hat{g}^{\star}\|_{2}^{2}\!+\!1),
\end{align*}
where the first inequality is due to Cauchy-Schwarz with
\begin{equation}\label{eq:dd_vs_real_model_mismatch}
    \Delta \pazocal{M}=\begin{bmatrix}
      \pazocal{V}_{0}^{L-1} & \pazocal{T}_{0}^{L-1} & \pazocal{C}_{0}^{L-1}
  \end{bmatrix}-\Phi,
\end{equation}
and the second and third inequalities are obtained by summing and subtracting $\hat{z}_{\mathrm{ini}}$ and $\hat{u}_f(\hat{g}^\star)$ to the first and second terms, then exploiting the definitions of and the bounds on $\varepsilon_{\mathrm{ini}}$ and $\varepsilon_u$. Combining all these results with \eqref{eq:bound_u}, we obtain 
\begin{align}\label{eq:bound_perf_cost}
     \nonumber J(g^\star)\leq~&  J(\hat{g}^\star)+c_1(\eta_{ini}+\|\hat{z}_{\mathrm{ini}}\|_{2}^{2})+c_2\eta_u+c_3\eta_g^2\\
     &+c_4\|\hat{g}^\star\|_{2}^{2}+2\bar{\varphi}\hat{\varepsilon}+8\bar{\varphi}\|\Delta\pazocal{M}\|_{2}^{2},
\end{align}
where
\begin{align*}
    c_1&=16\bar{\varphi}\|\Delta\pazocal{M}\|_{2}^{2},~~c_2=\bar{\varphi}(16\|\Delta\pazocal{M}\|_{2}^{2}+2),\\
    c_3&=16\bar{\varphi}\|Y_F\|_{2}^{2},~~~~~\!c_4=\bar{\varphi}(16\|\Delta\pazocal{M}\|_{2}^{2}\|\hat{U}_{F}\|_{2}^{2}+8\|\Delta Y_F\|_{2}^{2}),
\end{align*}
and $\hat{\varepsilon}$ in \eqref{eq:varepsilon_hat} is the 2-norm of the input/output tracking error achieved with the trajectories misclassified using misclassified data. We are now left to analyze the shrinkage terms in \eqref{eq:shrinkage_terms}. In particular, when considering the elastic net regularizer the following holds:
\begin{align*}
    r(g^\star)&=\lambda_1\|g^\star\|_{1}+\lambda_2\|g^\star\|_{2}^{2}\\
    &=\lambda_1\|g^\star+\hat{g}^\star-\hat{g}^\star\|_{1}+\lambda_2\|g^\star+\hat{g}^\star-\hat{g}^\star\|_{2}^{2}\\
    &\leq \lambda_1\|\varepsilon_g\|_1+\lambda_1\|\hat{g}^\star\|_{1}+2\lambda_2\|\varepsilon_g\|_{2}^{2}+2\lambda_2\|\hat{g}^\star\|_{2}^{2}\\
    &\leq r(\hat{g}^{\star})+\lambda_1 \eta_g+2\lambda_2\eta_g^2+\lambda_2\|\hat{g}^\star\|_{2}^{2}.
\end{align*}
Meanwhile, the CAP regularizer satisfies
\begin{align*}
    r(g^\star)&=\!\sum_{i=1}^{S}\lambda_i\|\pazocal{G}_i\|_{2}=\hat{r}(\hat{g}^\star)\!+\!\!\sum_{i=1}^{S}\!\left(\lambda_i\|\pazocal{G}_i^\star\|_{2}-\hat{\lambda}_i\|\hat{\pazocal{G}}_i^\star\|_{2}\right)\\
    &\leq \hat{r}(\hat{g}^\star)+\sum_{i=1}^{S}\lambda_i\|\pazocal{G}_i^\star\|_{2}\leq \hat{r}(\hat{g}^\star)+\bar{\lambda}\sum_{i=1}^{S}\|\pazocal{G}_i^\star\|_{2}\\
    & \leq \hat{r}(\hat{g}^\star)+\bar{\lambda}\sum_{i=1}^{S}\|\pazocal{G}_i^\star\|_{1}= \hat{r}(\hat{g}^\star)+\bar{\lambda}\|g^\star\|_{1}\\
    &=\hat{r}(\hat{g}^\star)+\bar{\lambda}\|g^\star+\hat{g}^\star-\hat{g}^\star\|_{1}\\
    & = \hat{r}(\hat{g}^\star)\!+\!\bar{\lambda}(\|\varepsilon_g\|_1+\|\hat{g}^{\star}\|_1)\leq \hat{r}(\hat{g}^\star)\!+\!\bar{\lambda}(\eta_g+\|\hat{g}^{\star}\|_1), 
\end{align*}
where $\lambda_i\neq \hat{\lambda}_i$ due to the dependence of the regularization parameters on the dimension of the data clusters, and $\bar{\lambda}=\max_{i \in [S]}\lambda_i$. The combination of these results with \eqref{eq:bound_perf_cost} concludes the proof.

\bibliographystyle{abbrv}
\bibliography{main.bib}

\end{document}